\newif\ifarxiv
\arxivtrue

\ifarxiv
\documentclass{article}
\usepackage{graphicx, subfig} 
\usepackage{geometry}
\usepackage{amsmath, amssymb, amsthm, bbm, bm, mathtools, natbib, xcolor}
\usepackage{hyperref}
\usepackage{wrapfig}
\usepackage[ruled,vlined]{algorithm2e}
\else
\documentclass{article}
\PassOptionsToPackage{numbers, compress}{natbib}
\usepackage{Styles/neurips_2025}
\usepackage[utf8]{inputenc} 
\usepackage[T1]{fontenc}    
\usepackage{hyperref}       
\usepackage{url}            
\usepackage{booktabs}       
\usepackage{amsfonts}       
\usepackage{nicefrac}       
\usepackage{microtype}      
\usepackage{xcolor}         
\usepackage{graphicx, subfig} 
\usepackage{geometry}
\usepackage{amsmath, amssymb, amsthm, bbm, bm, mathtools, natbib, xcolor}
\usepackage{hyperref}
\usepackage[ruled,vlined]{algorithm2e}
\usepackage{enumitem}
\usepackage{wrapfig}

\fi

\usepackage{sidecap}

\newcommand{\E}{\mathbb{E}}
\newcommand{\R}{\mathbb{R}}

\newcommand{\cA}{\mathcal{A}}

\newcommand{\cF}{\mathcal{F}}
\newcommand{\cH}{\mathcal{H}}
\newcommand{\ba}{\boldsymbol{a}}
\newcommand{\eps}{\varepsilon}

\renewcommand{\>}{\rangle}

\DeclareMathOperator*{\argmin}{arg\,min}

\newtheorem{theorem}{Theorem}[section]
\newtheorem{lemma}[theorem]{Lemma}

\newtheorem{definition}[theorem]{Definition}
\newtheorem{corollary}[theorem]{Corollary}
\newtheorem{example}[theorem]{Example}

\newenvironment{proofof}[1]{\par{\noindent \textit{Proof of #1.}}}{\hspace*{\fill} $\qed$ \par}

\newcommand{\mirah}[1]{\textcolor{magenta}{[Mirah: #1]}}

\newcommand{\temp}[1]{c^{\text{temp}}_{#1}}
\newcommand{\perm}[1]{c^{\text{perm}}_{#1}}
\newcommand{\permmean}[1]{c^{\text{perm-avg}}_{#1}}

\ifarxiv
\title{Algorithmic Aspects of Strategic Trading}
\author{Michael Kearns and Mirah Shi}
\date{\today}
\else
\title{Algorithmic Aspects of Strategic Trading}

\fi

\begin{document}

\maketitle

\begin{abstract}
Algorithmic trading in modern financial markets is widely acknowledged to exhibit strategic, game-theoretic behaviors whose complexity can be difficult to model. A recent series of papers~\citep{chriss2024optimal, chriss2024positionbuildingcompetitionrealworldconstraints, chriss2024competitiveequilibriatrading, chriss2025positionbuildingcompetitiongame} has made progress in the setting of trading for \emph{position building}. Here parties wish to buy or sell a fixed number of shares in a fixed time period in the presence of both temporary and permanent market impact, resulting in exponentially large strategy spaces. While these papers primarily consider the existence and structural properties of equilibrium strategies, in this work we focus on the algorithmic aspects of the proposed model. We give an efficient algorithm for computing best responses, and show that while the temporary impact only setting yields a potential game, best response dynamics do not generally converge for the general setting, for which no fast algorithm for (Nash) equilibrium computation is known. This leads us to consider the broader notion of Coarse Correlated Equilibria (CCE), which we show can be computed efficiently via an implementation of Follow the Perturbed Leader (FTPL). While we focus on equilibrium computation, our FTPL implementation learns no-regret strategies in any (adversarial) trading environment. 
We illustrate the model and our results with an experimental investigation, where FTPL
exhibits interesting behavior in different regimes of the relative weighting between temporary and permanent
market impact.
\end{abstract}


\section{Introduction}

There is both a vast commercial industry and a large quantitative finance literature centered on the problem of optimally executing trades in electronic exchanges under various conditions. Many brokerages and investment banks offer trading services tracking precise benchmarks such as the volume-weighted average price (VWAP) of a stock, or the prices obtained relative to the start of trading. Such services are both informed by and influence research in algorithmic trading (see Related Work below). The overarching goal in algorithmic trading is to acquire or sell a predetermined number of shares in a specified period of time\footnote{Such directives would typically come from higher-level constraints, such as the need to buy or sell shares of different stocks in order to maintain a portfolio that tracks a common index such as the S\&P 500, or from a hedge fund's quantitative model that detects and acts on perceived mispricings of assets.}, while minimizing the \emph{market impact} incurred by trading --- that is, the tendency of trading to push the asset price against the interests of the trader (buying causing prices to rise, selling causing prices to fall).

Despite the fact that trading in modern electronic markets has very obvious strategic aspects, and that traders informally incorporate game-theoretic considerations in their decisions and choices, it is rare to see such considerations explicitly modeled, primarily due to sheer complexity of doing so --- the possible strategies or algorithms are virtually innumerable, and there are a vast number of different exchange mechanisms and order types.

In a series of recent papers, Chriss~\citep{chriss2024optimal, chriss2024positionbuildingcompetitionrealworldconstraints, chriss2024competitiveequilibriatrading, chriss2025positionbuildingcompetitiongame} has made significant progress with the introduction and analysis of a stylized but realistic model for competitive trading. In Chriss' model, multiple traders play a game in which each player wishes to acquire either a long or short position in a common asset in a fixed time window, and each wishes to minimize their cost in doing so. As is standard in the finance literature, costs decompose into \emph{temporary} and \emph{permanent} market impact. Broadly speaking, temporary impact models the ``mechanical'' influence on prices inherent in the continuous double auction common in modern stock exchanges (due to worse prices as one eats further into the limit order books; see~\citep{nevmyvakaRL} for background), while permanent impact models the longer-term ``psychological'' effects of trading, such as perceptions of the value of the underlying asset. As we shall see in the model, at equilibrium, temporary impact tends to make traders want to avoid each other temporally, while permanent impact tends to make traders want to ``front run'' (trade before) each other. These competing forces, along with the fact that it may be beneficial to sell some shares en route to acquiring a net long position, make for a rich set of equilibrium strategies, which is the primary focus of Chriss' papers\footnote{Chriss' work is closely related to \citet{carlin}, who study equilibria in similar competitive trading setup.}.

In this work, we focus on the algorithmic and learning aspects of Chriss' model, and in doing so broaden the class of equilibria considered (specifically to coarse correlated equilibria (CCE), the class to which no-regret dynamics are known to converge). Our results include the following:
\begin{itemize}
    \item We prove that the algorithmic problem of computing the best-response trading schedule to the fixed actions of the other players admits an efficient dynamic programming solution.
    \item We show that if we consider temporary market impact only, then the game is a potential game and thus best-response dynamics will converge rapidly to a pure Nash equilibrium (NE).
    \item We show that (a slight variant of) permanent impact only is a zero-sum game, and thus the general game is a weighted mixture of a potential game and a zero-sum game. While this last statement is in fact true of {\em any game}, the particular mixture we obtain in our setting has interesting implications for our experimental results. We also observe that the general game is neither a potential nor zero-sum game. 
    \item Given that the general game seems to admit no special form, and thus computing Nash equilibria may be intractable, we next turn attention to computing (approximate) coarse correlated equilibria via no-regret dyanmics. We prove that despite the exponentially large strategy space, there is a computationally efficient implementation of Follow the Perturbed Leader (FTPL) for our game. In addition to their computational tractability, CCE are interesting in our setting due to the possibility of higher social welfare, and the suggestion that lightweight correlation at the exchanges themselves might render trading less costly. 
    Beyond equilibrium computation, our implementation of FTPL can be used to learn no-regret strategies in \textit{any} environment, where other traders could be acting adversarially.
    \item We conclude with an extensive experimental investigation of FTPL dynamics and CCE properties in different regimes of the relative weight on temporary and permanent impact. Despite the fact that there are no general guarantees about the CCE that FTPL will find in arbitrary games, the special structure of our aforementioned decomposition is reflected experimentally, with near-pure NE being found when temporary impact dominates and approximate mixed NE being found in the regime where permanent impact has twice the weight temporary impact. 
\end{itemize}

\subsection{Related Work}

The starting point for our work is the recent series of papers by Chriss introducing and analyzing a game-theoretic trading model~\citep{chriss2024optimal,chriss2024positionbuildingcompetitionrealworldconstraints,chriss2024competitiveequilibriatrading, chriss2025positionbuildingcompetitiongame}. Chriss begins by establishing the existence of Nash equilibria --- since he works in a continuous time and volume model, the pure strategy spaces are infinite, and thus existence does not immediately follow from Nash's celebrated theorem. Chriss imposes continuity and boundary conditions on strategies, which together allow him to prove existence. He then proceeds to examine equilibrium structure and to consider a number of variants of the model. 
Chriss' model is related to earlier work on optimal trade execution in a non-strategic setting~\citep{AlmgrenPortfolio}.
Here we consider a discrete time and volume version of Chriss' model, for which the pure strategy spaces are finite (though exponential in the time horizon) and thus mixed NE are guaranteed to exist. Since in reality trading must occur in discrete time steps and in whole shares, moving to a discrete model allows us to consider algorithmic issues more precisely, which is our primary interest.

Key to Chriss' model are standard notions of (temporary and permanent) market impact, on which there is a large literature;
see~\citep{Gatheral3Models,GatheralMI,Hautsch,Zarinelli,bouchaud,Webster} for a representative but partial sample. Broadly speaking, this literature considers various models for how trading activity influences asset prices, implications of those models for trading strategies, and empirical validation. A smaller body of work considers the algorithmic aspects of optimal trading, including machine learning approaches~\citep{evendarLimit,GanchevDark,nevmyvakaRL,kakadeVWAP}. Our work also focuses on algorithmic considerations, but in a game-theoretic setting.

Several other works consider trading games but operate in models different from ours. \citet{carlin} model players as either ``buyers" who wish to buy or sell a predetermined number of shares over a fixed time period or ``predators” seeking a zero net position, and they derive a pure Nash equilibrium. Similar to Chriss and our model, they use a linear model of temporary and permanent impact (with additional noise governed by a stochastic process). \citet{cont} study a sequential-move game, deriving Stackelberg equilibria when a resource-rich but less informed institutional trader ``leads," and an information-rich but limited-resource high frequency trader ``follows." Like Chriss, both \citet{carlin} and \citet{cont} work in a continuous time and volume model.

\section{Model and Preliminaries}

\ifarxiv
\paragraph{The trading game.}
\fi
We begin by describing the problem of strategic trading, where traders wish to build a position in a stock over a period of time. More precisely, we consider $n$ players, where every player $i\in[n]$ wishes to distribute purchases of $V_i$ shares of a stock over $T$ days (or other unit of time). Every player chooses a \textit{trading strategy} that specifies a trading schedule acquiring a target volume $V_i$ by day $T$. Here negative values of $V_i$ indicate a net short position, and positive values a net long position.

\begin{definition}[Trading Strategy]
    A trading strategy $a: [T] \to \mathbb{Z}$ is a mapping from a time step $t$ to the number of shares held by a player at time $t$, satisfying $a(0) = 0$ and $a(T) = V$. 
\end{definition}

A trading strategy can be equivalently described by the number of shares bought at every time step. Given a trading strategy $a$, we denote by $a'(t)$ the number of shares bought at time $t$---i.e. $a'(t) = a(t) - a(t-1)$---so that we can equivalently write $a(t) = \sum_{s=1}^t a'(s)$. We will interpret negative values of $a'(t)$ as the number of shares \textit{sold} at time $t$. We allow strategies that both buy and sell shares, regardless of the desired final net position $V$.

The action set $\cA(V_i)$ of a player $i$ is the set of all trading strategies $a$ satisfying $a(0)=0$ and $a(T)=V_i$. We will at times further restrict the action sets by implementing upper and lower trading limits---$\theta_U$ and $\theta_L$---bounding the number of shares that can be bought at any time step. We define the action set $\cA(V_i, \theta_L, \theta_U)$ as the set of all strategies $a$ additionally satisfying $\theta_L\leq a'(t) \leq \theta_U$. We will simply write $\cA_i$ when the parameters $V_i, \theta_L,$ and $\theta_U$ are not important to the discussion. We write $\ba  = (a_1,...,a_n) \in \prod_{i=1}^n \cA_i$ to denote an action profile and $\cA_{-i} = \prod_{j\neq i} \cA_j$ to denote the action space of all players excluding player $i$.

A player's cost per share purchased will take into account two basic sources of market impact --- \textit{temporary impact} and \textit{permanent impact}. 

\begin{definition}[Temporary Impact Cost]
    The temporary impact cost of a trading strategy $a_i\in\cA_i$ against strategies $a_{-i}\in\cA_{-i}$ is:
    \[
    c^{\text{temp}}(a_i, a_{-i}) = \sum_{t=1}^T a'_i(t) \sum_{j=1}^n a'_j(t)
    \]
\end{definition}

\begin{definition}[Permanent Impact Cost]
    The permanent impact cost of a trading strategy $a_i\in\cA_i$ against strategies $a_{-i}\in\cA_{-i}$ is:
    \[
    c^{\text{perm}}(a_i, a_{-i}) = \sum_{t=1}^T a'_i(t) \sum_{j=1}^n a_j(t-1)
    \]
\end{definition}

In other words, temporary impact considers the number of \textit{instantaneous} shares bought/sold by all players at any time step, while permanent impact considers the number of shares bought/sold by all players \textit{prior to} that time step. The cost is then formulated as a linear function of market impact. Following \citet{chriss2024optimal}, we define a player's general cost as the sum of their temporary impact cost and their permanent impact cost. We will be able to control the relative contributions of temporary and permanent impact costs via a \textit{market impact coefficient} $\kappa$. 

\begin{definition}[Cost of Trading]
    Fix $\kappa\geq 0$. The cost of a trading strategy $a_i\in\cA_i$ against strategies $a_{-i}\in\cA_{-i}$ is given by:
    \[
    c(a_i, a_{-i}) = \sum_{t=1}^T \left(a'_i(t) \sum_{j=1}^n a'_j(t) + \kappa \cdot a'_i(t) \sum_{j=1}^n a_j(t-1) \right)
    \]
\end{definition}

It is worth noting that this model can be ``explained'' in terms of assumptions on the underlying 
limit order book dynamics that mediate all trading activity. More specifically, considering (without
loss of generality) only
players who wish to buy shares to obtain a long position, if we assume that (a) the distribution of
share prices in the sell order book is uniform, and that (b) once consumed, shares in the sell book
are never replenished by the arrival of new shares, then we recover Chriss' model. Assumption (a)
corresponds to his linear temporary cost model, and assumption (b) corresponds to his linear
permanent cost model. Assumption (b) can be viewed as an extreme form of permanent impact, in
that any trading activity that drives the price up will never be reversed --- the market always
revalues the security at the current price level. It is then possible to interpret $\kappa$
as a \emph{liquidity replenishment} parameter, in that intermediate values of $\kappa$ model
new sell orders arriving at previous price levels at some rate. \citet{carlin} attribute permanent impact to a similar phenomenon---that decreases in supply cause increases in price---but do not directly map it onto limit order book dynamics.

While both of these assumptions are stylized and
somewhat unrealistic in practice, they at least ground our model in assumptions about the low-level dynamics
of the exchanges. They also point to more realistic variants of the model, in which we assume
more natural price distributions in the order books (for instance, it is common for much more
liquidity to be aggregated near the bid and ask prices, and to thin out away from them), and
less extreme replenishment assumptions.

Given a fixed action profile, the best response of player $i$ is a strategy that minimizes cost.

\begin{definition}[Best Response]
    Consider a player $i$ with action set $\cA(V_i, \theta_L, \theta_U)$ and cost function $c$. The best response of player $i$ to action profile $a_{-i}$ is the action $a_i^* = \argmin_{a\in\cA_i} c(a, a_{-i})$.
\end{definition}

\ifarxiv
\paragraph{Equilibria Concepts.}
\fi
We will study several basic equilibria concepts, defined below in increasing generality. 

\begin{definition}[Pure Nash Equilibrium]
    An action profile $\ba$ is an $\eps$-approximate pure Nash equilibrium if for all $i$, $c(a_i,a_{-i})\leq \min_{a\in\cA_i}c(a,a_{-i}) + \eps$. When $\eps=0$, $\ba$ is a pure Nash equilibrium.
\end{definition}

\begin{definition}[Mixed Nash Equilibrium]
    A profile of (independent) distributions $\mathbf{D} = (D_1\times...\times D_n) \in \prod_{i=1}^n \Delta \cA_i$ is an $\eps$-approximate mixed Nash equilibrium if for all $i$, $\E_{\ba\sim\mathbf{D}}[c(a_i,a_{-i})] \leq \min_{a\in\cA_i}\E_{\ba\sim\mathbf{D}}[c(a,a_{-i})] + \eps$. When $\eps=0$, we say $\mathbf{D}$ is a mixed Nash equilibrium.
\end{definition}

\begin{definition}[Correlated Equilibrium (CE)]
    A distribution $\mathbf{D}$ over action profiles is an $\eps$-approximate correlated equilibrium if for all $i$, for all swap functions $\phi_i:\cA_i\to\cA_i$, $\E_{\ba\sim\mathbf{D}}[c(a_i,a_{-i})] \leq \E_{\ba\sim\mathbf{D}}[c(\phi(a_i),a_{-i})] + \eps$. When $\eps=0$, $\mathbf{D}$ is a correlated equilibrium.
\end{definition}

\begin{definition}[Coarse Correlated Equilibrium (CCE)]
    A distribution $\mathbf{D}$ over action profiles is an $\eps$-approximate coarse correlated equilibrium if for all $i$, $\E_{\ba\sim\mathbf{D}}[c(a_i,a_{-i})] \leq \min_{a\in\cA_i} \E_{\ba\sim\mathbf{D}}[c(a,a_{-i})] + \eps$. When $\eps=0$, we say $\mathbf{D}$ is a coarse correlated equilibrium.
\end{definition}

\ifarxiv
\paragraph{Examples.}
To provide some intuition of the game, we give some examples of equilibria strategies, under varying market impact coefficients $\kappa$. Recall that $\kappa$ determines the relative contributions of temporary and permanent impact. Figure \ref{fig:ex-buy-only} shows pure Nash equilibria strategy pairs for the \textit{buy-only} setting (i.e. $\theta_L=0$). In this setting, we see a clear tension between temporary and permanent impact; when players pay only temporary impact cost (i.e. $\kappa=0$), the tendency is to spread out trading activity to avoid the opponent---and themselves. When players pay only permanent impact cost, the tendency is to trade ahead (in fact, buying everything at $t=1$ incurs $0$ permanent impact cost in our model). Here, $\kappa=2$ is large enough to induce this behavior. For the intermediary case of $\kappa=1$, players strike a balance between the two.

\begin{figure}
\centering
\begin{tabular}{cccc}
& \includegraphics[width=35mm,trim={0mm 4mm 3mm 0},clip]{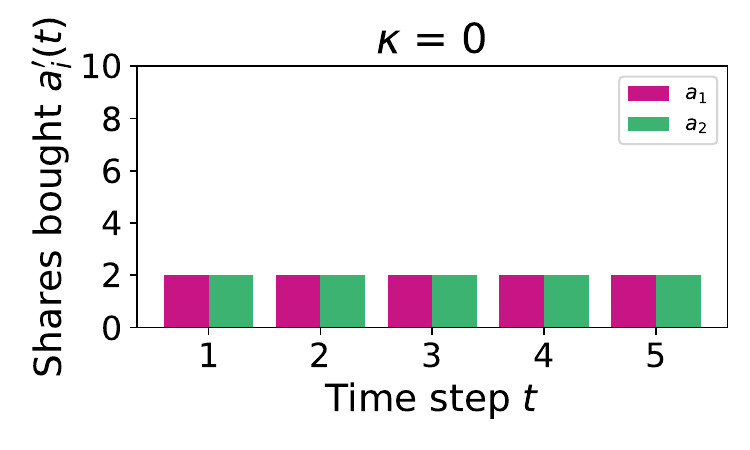} & 

\includegraphics[width=33mm,trim={12mm 4mm 3mm 0},clip]{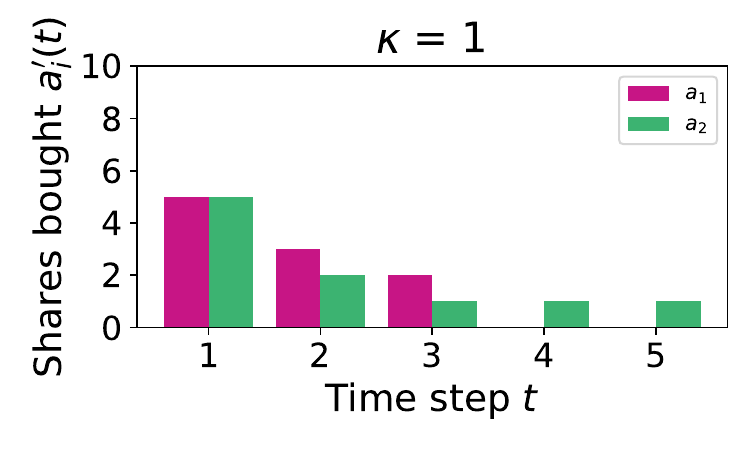}  
& \includegraphics[width=33mm,trim={12mm 4mm 3mm 0},clip]{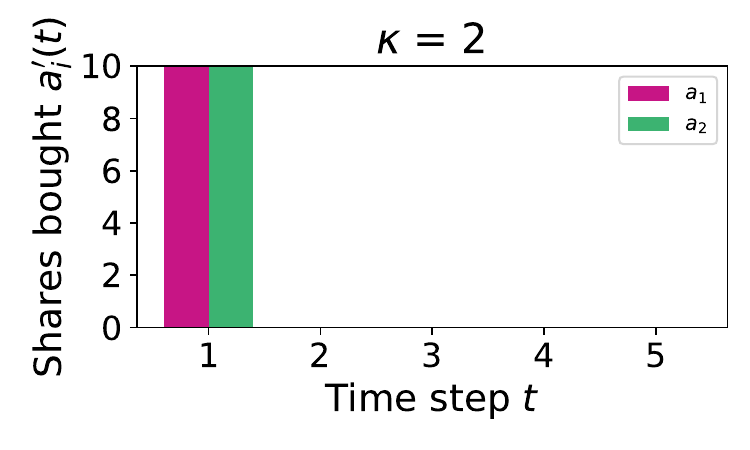} 
\end{tabular}
\caption{Pure Nash equilibria strategies $(a_1, a_2)$ in the trading game with two players, for different $\kappa$. In this example, $T=5$, $V_1=V_2=10$, $\theta_U=10$, and $\theta_L=0$.}
\label{fig:ex-buy-only}
\end{figure}

\textit{Selling} complicates the picture. For instance, buying upfront was previously the best \textit{buy-only} strategy for large enough $\kappa$. When selling is allowed, players tend to want to sell immediately after their opponent buys (when costs are high) and buy immediately after their opponent sells (when costs are low). Figure \ref{fig:ex-sell} gives examples of best response strategies exhibiting this behavior.

\begin{SCfigure}
\centering
\begin{tabular}{cccc}
& \includegraphics[width=35mm,trim={8mm 4mm 3mm 0},clip]{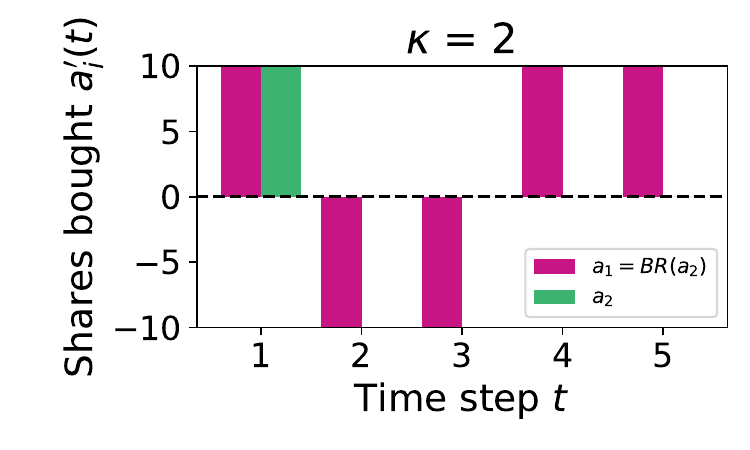} & \includegraphics[width=33mm,trim={17mm 4mm 3mm 0},clip]{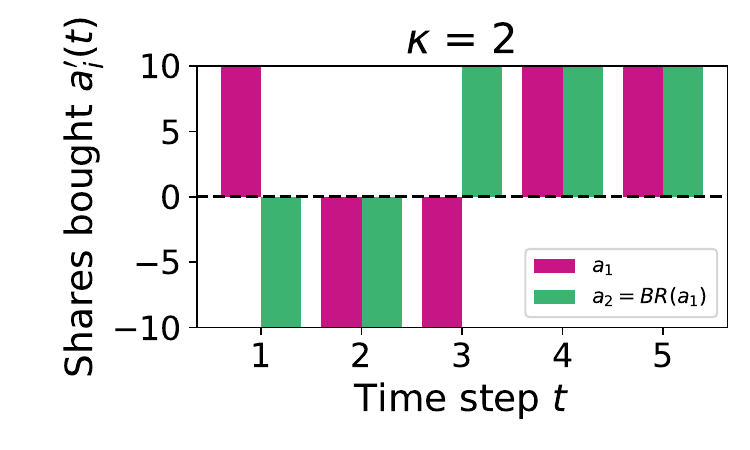}
\end{tabular}
\caption{Examples of best response strategies under $\kappa=2$ when allowing for selling. On the left, $a_1$ is a best response to $a_2$; on the right, $a_2$ is a best response to $a_1$. Here, $T=5$, $V_1=V_2=10$, $\theta_U=10$, and $\theta_L=-10$.}
\label{fig:ex-sell}
\end{SCfigure}
\fi

\section{Computing a Best Response}\label{sec:BR}

We begin by giving a dynamic programming algorithm that finds a best response to any profile of trading strategies. The dynamic programming algorithm we present in this section (Algorithm \ref{alg:dp-BR}) will later serve as a building block in our algorithms for equilibrium computation.

\begin{algorithm}[t]
    \KwIn{Target volume $V$, trading limits $\theta_L, \theta_U$, one-step cost function $p^t$}
    \KwOut{Best response $a^*  = \argmin_{a\in \cA(V,\theta_L, \theta_U)} \phi^T(a)$, where $\phi^T(a) = \sum_{t=1}^T p^t(a(t-1), a'(t))$}
    
    \ifarxiv
    \vspace{0.5em}
    \fi


    \For{$s=V-\theta_U t$ \KwTo $V-\theta_L t$}{
        
        Initialize $\mathrm{OPT}(T, s) = p^t(V-s, s)$
        
        Initialize $\mathrm{BR}(T, s) = s$
    }

    \For{$t=T-1$ \KwTo $1$}{
        \For{$s=V-\theta_U t$ \KwTo $V-\theta_L t$}{
            Let $ \mathrm{OPT}(t, s) = \min_{\theta_L\leq k \leq \theta_U} \left( \mathrm{OPT}(t+1, s-k) + p^t(V-s, k) \right) $ 
            and $\mathrm{BR}(t, s) = \argmin_{\theta_L\leq k \leq \theta_U} \left( \mathrm{OPT}(t+1, s-k) + p^t(V-s, k) \right)$
        }
    }

    \For{$t=1$ \KwTo $T$ \tcp{recover the optimal strategy $a^*$} }{ 
        Let ${a^*}'(t) = \mathrm{BR}(t, V)$\;
        Let $V = V - {a^*}'(t)$\;
    }
    Return $a^*$ \;

    \caption{Best response over one-step costs (\textsf{BR})}
    \label{alg:dp-BR}
\end{algorithm}

The key insight is that the optimal strategy beginning at any time $t$ depends only on the permanent impact of prior trading activity---which is determined by the number of shares held prior to $t$. For a strategy that must hold $V$ shares at time $T$, this can be expressed as the number of remaining shares that need to be bought. That is, if $m$ shares are held before $t$, then $V-m$ shares must be bought from $t$ onwards. Thus, if we knew the number of remaining shares to be bought, we can define a subproblem solving for the optimal strategy beginning at time $t$ with $s$ shares remaining. In short, then, our dynamic programming algorithm solves for the best response inductively over time steps and remaining shares. The number of inductive steps will determine the computation required---which we will see depends polynomially on $\theta_L, \theta_U$, and $T$.

This idea lends itself to more general cost minimization problems over trading strategies; in fact, we will present and analyze a more general form of the algorithm than what is required to compute a best response in the trading game. In particular, we show how to compute 
$
a^* = \argmin_{a\in\cA_i} \phi^T(a)
$
for any cost function $\phi^T(a)$ that can be written as the sum of \textit{one-step costs} $p^t$ that depend on the ``state" of trades at time $t$: 
$$\phi^T(a) = \sum_{t=1}^T p^t(a(t-1), a'(t))$$ 
We remark that $p^t$ is defined to take as input $a(t-1)$ rather than $a(t)$ simply for ease of presentation later on --- the quantities $a(t-1)$ and $a(t)$ are essentially interchangeable, since $a(t) = a(t-1) + a'(t)$. 

\ifarxiv
Our trading cost can be written in this way. Specifically, we can decompose the cost $c(a_i,a_{-i})$ of a strategy $a_i$ into one-step costs, parameterized by $a_{-i}$ and $\kappa$: $$c^t(a_i(t-1), a_i'(t); a_{-i}, \kappa) = a'_i(t) \sum_{j=1}^n a'_j(t) + \kappa \cdot a'_i(t) \sum_{j=1}^n a_j(t-1)$$ so that $c(a_i,a_{-i}) = \sum_{t=1}^T c^t(a_i(t-1), a_i'(t); a_{-i}, \kappa)$. Thus, instantiating our algorithm with $p^t = c^t$ computes a best response in the trading game. 
\fi
We will later rely on this general form of the algorithm in Section \ref{sec:no-regret-dynamics}, when we reduce to a cost minimization problem that can be formulated as the sum of one-step costs. 

\begin{theorem}\label{thm:dp}
    Given a target volume $V$, trading limits $\theta_L,\theta_U$, and a cost function $\phi^T$ such that $\phi^T(a) = \sum_{t=1}^T p^t(a(t-1), a'(t))$ for some function $p^t$, Algorithm \ref{alg:dp-BR} computes $a^*  = \argmin_{a\in \cA(V,\theta_L, \theta_U)} \phi^T(a)$. Moreover, Algorithm \ref{alg:dp-BR} runs in time $O((\theta_U-\theta_L)^2T^2)$.
\end{theorem}
\ifarxiv
The algorithm is simple to describe. Let $\mathrm{OPT}(t, s)$ be the minimum cost to buy $s$ shares beginning at time $t$. First, observe that at the last time step $T$, a strategy must buy all remaining shares $s$. Furthermore, if $s$ shares remain, it must be that $V-s$ shares are held before $T$. Thus, $\mathrm{OPT}(T, s) = p^t(V-s, s)$. Now we work backwards. We can compute: $$\mathrm{OPT}(t, s) = \min_k [ \mathrm{OPT}(t+1, s-k) + p^t(V-s, k) ]$$ Why? The one-step cost of buying $k$ shares at time $t$ with $s$ shares remaining is $p^t(V-s, k)$; the minimum remaining cost is simply the minimum cost of buying $s-k$ shares beginning at the next time step, i.e. $\mathrm{OPT}(t+1, s-k)$. The cost of a best response strategy is then $\mathrm{OPT}(1, V)$; some simple bookkeeping will allow us to recover the optimal strategy.

We present these ideas in more detail below.
\begin{proofof}{Theorem \ref{thm:dp}}
    We first show that Algorithm \ref{alg:dp-BR} finds a best response. As above, let $\mathrm{OPT}(t, s)$ be the minimum cost for a strategy to buy $s$ shares beginning at time $t$. The optimal cost is therefore $\mathrm{OPT}(1, V)$. Let $a$ denote any strategy in $\cA(V, \theta_L, \theta_U)$. Since $a$ satisfies $\theta_L\leq a'(t) \leq \theta_U$ for all $t$, we have that $\theta_L t\leq a(t) \leq \theta_U t$ and thus, $V-\theta_Ut \leq s \leq V-\theta_Lt$ for all $t$. 
    
    We now proceed via induction. If $s$ shares remain at the last time step, then it must be that $a(T-1) = V-s$ and $a'(T)=s$. Thus we can define the base case $\mathrm{OPT}(T, s)$ as the cost of buying $s$ shares at time $T$, given that $V-s$ shares are held up until time $T$, i.e.:
    \[
    \mathrm{OPT}(T, s) = p^t(V-s, s)
    \]

    The inductive step rests on the following fact: for $t=1,...,T-1$ and $s = V-\theta_Ut,...,V-\theta_Lt$, 
    \begin{align*}
        \mathrm{OPT}(t, s) = \min_{\theta_L\leq k \leq \theta_U} \left( \mathrm{OPT}(t+1, s-k) + p^t(V-s, k) \right)
    \end{align*} 
    To see this, observe that if $s$ shares remain at time $t$, then it must be that $V - s$ shares are held up until time $t$, i.e. $a(t-1) = V-s$. Suppose $k$ shares are bought at time $t$---i.e. $a'(t)=k$. Then, the one-step cost incurred at time $t$ is precisely $p^t(V-s, k)$. The optimal remaining cost is the minimum cost to buy the remaining $s-k$ shares beginning at time $t+1$---that is, $\mathrm{OPT}(t+1, s-k)$. The optimal solution to buy $s$ shares beginning at time $t$ buys some number of shares $k \in [\theta_L,\theta_U]$ at time step $t$. Since the inductive step chooses $k$ to minimize the cost beginning at time $t$, it is optimal. This proves the inductive step.

    Now, for every $t, s$ pair, Algorithm \ref{alg:dp-BR} stores the minimizer of the previous expression: $$\mathrm{BR}(t,s) = \argmin_{\theta_L\leq k \leq \theta_U} \left( \mathrm{OPT}(t+1, s-k) + p^t(V-s, k) \right)$$ Thus, backtracking starting at $\mathrm{\mathrm{BR}}(1,V)$ recovers the number of shares to buy at every time $t$ in an optimal solution to buy $V$ shares starting at $t=1$, and so recovers an optimal strategy $a^*(t)$.

    It remains analyze the running time of Algorithm \ref{alg:dp-BR}. There are 3 nested iterations. The first iterates through each $t\in[T]$. For each $t$, it iterates through all values $s$ between $V-\theta_Ut$ and $V-\theta_Lt$. Then for each value of $s$, it finds $\mathrm{OPT}(t,s)$ by iterating through all values $k$ between $\theta_L$ and $\theta_U$. Recovering the optimal strategy takes time $T$. Thus, the running time is:
    \begin{align*}
        \left(\sum_{t=1}^T (V-\theta_Lt-(V-\theta_Ut)) (\theta_U-\theta_L)\right) + T 
        &= \left((\theta_U-\theta_L) \sum_{t=1}^T (\theta_Ut-\theta_Lt)\right) + T\\
        &= \left((\theta_U-\theta_L)^2 \sum_{t=1}^T t\right) + T\\
        &= \left((\theta_U-\theta_L)^2 \cdot \frac{T(T+1)}{2}\right) + T\\
        &= O((\theta_U-\theta_L)^2T^2)
    \end{align*}
    which proves the theorem. 
\end{proofof}
\fi

\section{A Decomposition of the Trading Game}\label{sec:decomposition}

In this section, we give a decomposition of the trading game that will have implications for equilibrium computation. At a high level, we show that the trading game is a mixture of a potential game and a constant-sum game. While this is the case for any game
\ifarxiv
\footnote{Any game is the sum of a potential game and zero-sum game (private communication from Aaron Roth): Consider an arbitrary game where player $i\in[n]$ has cost function $c_i$. We can decompose this into the sum of two games where every player $i$ has cost $c_{i,1}=\frac{\sum_{i=1}^n c_i}{n}$ in the first, and $c_{i,2} = \frac{(n-1)c_i-\sum_{j\neq i} c_j}{n}$ in the second. In the first game, every player has the same cost, and so the sum of costs $\sum_{i=1}^n c_{i,1}$ is a potential function. In the second game, the sum of costs $\sum_{i=1}^n c_{i,2} = 0$, since $(n-1) \sum_{i=1}^n c_i = \sum_{i=1}^n \sum_{j\neq i} c_j$, and so is zero-sum.}
\fi
, we show that, interestingly, the potential game arises from trading under temporary impact only, while the constant-sum game arises essentially from trading under permanent impact only. 
\ifarxiv
More precisely, we will show that the trading cost decomposes into:
\[
c(a_i, a_{-i}) = \left(1-\frac{\kappa}{2}\right)\cdot\temp{}(a_i, a_{-i}) + \kappa\cdot\permmean{}(a_i, a_{-i})
\]
where $\temp{}$ defines a potential game and $\permmean{}$ (a slight modification of permanent impact cost) defines a constant-sum game.
\fi

\ifarxiv
This decomposition shows that the basic structure of the trading game differs with underlying market impact\footnote{Empirical studies suggest that the ratio of temporary to permanent impact varies significantly in markets \citep{sadka, carlin}; our result highlights how incentives can vary with this ratio.}. Most saliently, when $\kappa=0$ (i.e. when players face only temporary impact), the game is a potential game, and so simple ``best-response dynamics" converge to a pure Nash equilibrium. When $\kappa=2$ (i.e. the contribution of permanent impact is twice that of temporary impact), the game is constant-sum. In this case, the empirical history of ``no-regret dynamics" converges to a mixed Nash equilibrium \citep{Cesa-Bianchi_Lugosi_2006}. For all other $\kappa$, the game is a weighted mixture of a potential game and a constant-sum game; the value of $\kappa$ determines its proximity to either. Later on, we will see that this basic structure is reflected in interesting ways in our experimental evaluations.
\fi

\ifarxiv
The remainder of the section is dedicated to substantiating these ideas. We begin by separately analyzing the terms in the decomposition --- this coincides with thinking separately about the temporary impact only and permanent impact only regimes. In subsection \ref{subsec:temp-impact}, we focus on the temporary impact only regime: we show that $\temp{}$ defines a potential game and provide a simple learning dynamic---best response dynamics---that converge to a pure Nash equilibrium. 
In subsection \ref{subsec:perm-impact}, we focus on the permanent impact only regime: we define the variant of permanent impact cost $\permmean{}$ and show that it is constant-sum. Finally, in subsection \ref{subsec:decomp}, we prove the decomposition of the general trading game. 
\fi

\subsection{The Temporary Impact Regime}\label{subsec:temp-impact}

\ifarxiv
Recall that the temporary impact cost is summarized by the instantaneous number of shares bought/sold:
\[
    c^{\text{temp}}(a_i, a_{-i}) = \sum_{t=1}^T a'_i(t) \sum_{j=1}^n a'_j(t)
\]
\fi
As mentioned in the Introduction, temporary impact can be viewed as modeling the mechanical aspects of trading in the double-auction order book mechanism of modern electronic markets, in which queues or ``books'' of buy and sell orders are ordered by price, and (say) a buyer demanding immediately liquidity must consumer successive orders with increasing prices in the sell book.

We show that the game defined by $c^{\text{temp}}$ is a potential game, and so, by the classical result of \citet{monderer96potential}, the simple procedure of \textit{best response dynamics} converges to a pure Nash equilibria.  

\begin{definition}[Potential Game]
    Consider an $n$-player game $G$ where each player $i\in[n]$ chooses actions from an action set $\cA_i$ and has cost function $c_i$. $G$ is an exact potential game if there exists a potential function $\phi: \cA_1 \times ... \times \cA_n \to \mathbb{R}$ such that for all players $i \in [n]$ and actions $a_i, b_i \in \cA_i$, 
    \[
    \phi(b_i, a_{-i}) - \phi(a_i, a_{-i}) = c_i(b_i, a_{-i}) - c_i(a_i, a_{-i})
    \]
\end{definition}

\begin{theorem}\citep{monderer96potential}\label{thm:potential}
    In any finite potential game, best response dynamics converges to a pure Nash equilibrium. 
\end{theorem}

In best response dynamics, players move sequentially to a beneficial deviation, as long as the strategy profile is not a pure Nash equilibrium. The rationale behind Theorem \ref{thm:potential} is simple. Any deviation strictly decreases a player's cost, and thus the potential function. Since the game is finite, the potential function must reach a minimum, at which point it must be that are no beneficial deviations---that is, players reach a Nash equilibrium. Given this fact, it suffices to produce a potential function for the temporary impact only setting. 


\begin{theorem}\label{thm:temp-potential}
    Consider an instance of the trading game where for every player $i\in[n]$, the cost of strategy $a_i$ against strategies $a_{-i}$ is given by the temporary impact cost $c^{\text{temp}}(a_i, a_{-i})$, i.e. $\kappa=0$. This is a potential game with potential function:
    \[
    \phi(\ba) = \sum_{t=1}^T \sum_{i=1}^n a_i'(t) \sum_{j\geq i} a_j'(t)
    \]
\end{theorem}
\ifarxiv
\begin{proof}
    The task is to show that the change in $\phi$ exactly measures the change in temporary impact cost resulting from a unilateral deviation. 
    For ease of notation, let's define $h(a_i, a_j) \coloneq \sum_{t=1}^T a_i'(t)a_j'(t)$. And so we can write $c^{\text{temp}}(a_i, a_{-i}) = \sum_{j=1}^n h(a_i, a_j)$ and $\phi(\ba) = \sum_{i=1}^n \sum_{j\geq i} h(a_i, a_j)$.

    Suppose player $k$ deviates from $a_k$ to $b_k$. Since $h$ is symmetric, i.e. $h(a_i, a_j) = h(a_j, a_i)$, we can write the change in potential as:
    \begin{align*}
        \phi(b_k, a_{-k}) - \phi(a_k, a_{-k}) &= \sum_{j=1}^n h(b_k, a_j) + \sum_{i\neq k} \sum_{j\geq i, j\neq k} h(a_i, a_j) - \sum_{j=1}^n h(a_k, a_j) - \sum_{i\neq k} \sum_{j\geq i, j\neq k} h(a_i, a_j) \\
        &= \sum_{j=1}^n h(b_k, a_j) - \sum_{j=1}^n h(a_k, a_j) \\
        &= c^{\text{temp}}(b_k, a_{-k}) - c^{\text{temp}}(a_k, a_{-k})
    \end{align*}
    That is, all terms not involving $k$ cancel out, and the remaining $n$ terms involving $k$ exactly match the change in cost. This proves the theorem.
    

\end{proof}
\fi

\ifarxiv
We have thus established that a pure Nash equilibrium exists in the temporary impact only setting and can be found using best response dynamics --- but how quickly? 
\fi
To answer this, it is common to settle for an $\eps$-approximate equilibrium. We can bound the number of rounds best response dynamics will run for, as long as each deviation leads to a large enough improve in cost---at least $\eps$. That is, as long as the strategy profile is not an $\eps$-approximate Nash equilibrium, players sequentially move to a strategy that lowers their cost by at least $\eps$. Notice we can implement best response dynamics using Algorithm \ref{alg:dp-BR} to sequentially compute best responses, and verifying if it gives a large enough improvement.

\ifarxiv
\begin{algorithm}[t]
    \KwIn{Target volumes $V_1,...,V_n$, trading limits $\theta_L,\theta_U$}
    \KwOut{$\eps$-approximate Nash equilibrium $\ba$}
    
    Initialize $\ba=(a_1,...,a_n)$ arbitrarily.

    Define $c^{\text{temp, t}}(a_i(t-1), a_i'(t); a_{-i}) = a_i'(t)\sum_{j=1}^n a'_j(t)$ to be the one-step temporary cost\;

    \For{$i=1$ \KwTo $n$}{
        Let $\tilde{a}_i \gets \textsf{BR}(V_i, \theta_L, \theta_U, c^{\text{temp, t}}(a_i(t-1), a_i'(t); a_{-i})$)\; 
        If $c^{\text{temp}}(\tilde{a}_i, a_{-i}) \leq c^{\text{temp}}(a_i, a_{-i}) - \eps$, set $a_i \gets \tilde{a}_i$\;
    }
    
    Return $\ba$
    
    \caption{$\eps$-approximate Nash equilibrium (temporary impact only)}
    \label{alg:BR-dynamics}
\end{algorithm}
\fi

\begin{theorem}\label{thm:br-dynamics}
    Best response dynamics (Algorithm \ref{alg:BR-dynamics}) returns an $\eps$-approximate Nash equilibrium in the temporary impact only setting. Moreover, it has running time bounded by $O\left( \frac{n^3\theta^4T^3}{\eps} \right)$, where $\theta = \max\{|\theta_L|, |\theta_U|\}$.
\end{theorem}
\ifarxiv
\begin{proof}
    By definition, the strategy profile found by the algorithm is an $\eps$-approximate Nash equilibrium. Now, by Theorem \ref{thm:temp-potential}, we have that for any player $i$, $c^{\text{temp}}(a_i, a_{-i}) - c^{\text{temp}}(\tilde{a}_i, a_{-i}) = \phi(a_i,a_{-i}) - \phi(\tilde{a}_i,a_{-i})$, where $\phi(\ba) = \sum_{t=1}^T \sum_{i=1}^n a_i'(t) \sum_{j\geq i} a_j'(t)$ is the potential function. Thus, for every deviation from $a_i$ to $\tilde{a}_i$, $ \phi(a_i,a_{-i}) - \phi(\tilde{a}_i,a_{-i}) \geq \eps$. And so to bound the running time, it suffices to bound the magnitude of $\phi$. Since $|a_i'(t)|\leq \theta$ for all players $i$ and time steps $t$, we can calculate for any $\ba$:
    \begin{align*}
        |\phi(\ba)| = \left| \sum_{t=1}^T \sum_{i=1}^n a_i'(t) \sum_{j\geq i} a_j'(t) \right| \leq \sum_{t=1}^T \sum_{i=1}^n \sum_{j\geq i} \theta^2 = \frac{n(n+1)T\theta^2}{2}
    \end{align*}
    Therefore the algorithm halts after at most $\frac{2n(n+1)T\theta^2}{2\eps} = \frac{n(n+1)T\theta^2}{\eps} $ deviations. Each deviation is found using at most $n$ calls to Algorithm \ref{alg:dp-BR}. Thus, plugging in the guarantees of Algorithm \ref{alg:dp-BR} (Theorem \ref{thm:dp}) bounds the total running time. 
\end{proof}
\fi

We conclude this discussion by verifying that the general trading game is \textit{not} a potential game; we show that best response dynamics can cycle in the presence of both temporary and permanent impact. The intuition is that temporary and permanent impact can create counteracting forces. As we previously discussed, temporary impact causes players to want to spread out their trades so as to avoid buying many shares at any time step. On the other hand, permanent impact causes players to want to trade ahead of everyone else. The tension between the two behavior---spreading out trades and trading ahead---can cause best response dynamics to oscillate.

\begin{theorem}\label{thm:not-potential}
    The general trading game is not a potential game. 
\end{theorem}
\ifarxiv
\begin{proof}
    Theorem \ref{thm:potential} tells us that for any finite potential game, best response dynamics is guaranteed to converge. Thus it suffices to give an instance for which best response dynamics does not converge. Below we show an instance with $T=5$, $\kappa=1$, and two players, both with the action set $\cA(V)$ for $V=5$.  

    Consider a run of best response dynamics, where player 1's strategy $a_1$ is initialized to be: $$a_1(1)=2, a_1(2)=2, a_1(3)=1, a_1(4)=0, a_1(5)=0$$ and player 2's strategy $a_2$ is initialized to be:
    $$a_2(1)=1, a_2(2)=1, a_2(3)=1, a_2(4)=1, a_2(5)=1$$ 
    Now, player 2 can decrease his cost against $a_1$ by playing $a_2'$, where:
    $$a_2'(1)=3, a_2'(2)=1, a_2'(3)=0, a_2'(4)=0, a_2'(5)=1$$ We have that $c(a_2, a_1) = 36$ while $c(a_2', a_1) = 33$. Then, player 1 can decrease her cost against $a_2'$ by playing $a_1'$, where: $$a_1'(1)=2, a_1'(2)=1, a_1'(3)=1, a_1'(4)=1, a_1'(5)=0$$ We have that $c(a_1, a_2') = 35$ while $c(a_1', a_2') = 34$. Then, player 2 can decrease his cost against $a_1'$ by playing $a_2''$, where:
    $$a_2''(1)=2, a_2''(2)=2, a_2''(3)=1, a_2''(4)=0, a_2''(5)=0$$ We have that $c(a_2', a_1') = 32$ while $c(a_2'', a_1') = 31$. Note that $a_2'' = a_1$, and so best response dynamics will cycle, i.e. it will not converge. This completes the proof.    
\end{proof}
\fi

\subsection{The Permanent Impact Regime}\label{subsec:perm-impact}

Next, turning to the permanent impact only setting, we show that permanent impact essentially induces a constant-sum game. More accurately, while permanent impact cost $c^{\text{perm}}$ is \textit{not} constant-sum, a semantic-preserving variant of $c^{\text{perm}}$ \textit{is}. 
\ifarxiv
The following example demonstrates that $c^{\text{perm}}$ is not constant-sum.

\begin{example}
Consider the trading game with two players. Suppose both players buy all $V=V_1=V_2$ shares at $t=1$ and 0 shares at every step
afterwards. Then the sum of permanent impact costs for both players is $0$. On the
other hand, suppose both players buy $V/T$ shares at each time step. Then the sum of permanent impact 
costs is $$2 \sum_{t=1}^T \frac{V}{T}\left(\frac{2V}{T}\cdot(t-1)\right) = \frac{4V^2}{T^2}\sum_{t=1}^T t - \frac{4V^2}{T} = \frac{4V^2}{T^2}\cdot\frac{T(T+1)}{2} - \frac{4V^2}{T} = 2V^2 - \frac{2V^2}{T}$$
which approaches $2V^2$ as $T$ becomes large. Thus the permanent impact only setting is not constant-sum.\footnote{In fact, using the same example, we can show that the general game is not constant-sum. If both players buy all $V$ shares upfront, the sum of temporary impact costs for both players is $4V^2$, and so the sum of temporary and permanent impact costs is $4V^2$. On the other hand, if both players buy $V/T$ shares at each time step, then the sum of temporary costs is $2\sum_{t=1}^T (V/T)(2V/T) = 4V^2/T$, which approaches 0 as $T$ becomes large. So the sum of temporary and permanent impact costs approaches $\kappa\cdot 2V^2$ as $T$ becomes large. Thus the general game is not zero-sum for $\kappa\neq 2$.} 
\end{example}

Now, we consider a slight variant of permanent impact cost that is in fact constant-sum
\ifarxiv
\footnote{Andrew Bennett, private communication.}.
\else
\footnote{Private communication, anonymous colleague.}.
\fi
\fi
We define $$\permmean{}(a_i,a_{-i}) \coloneqq \frac{1}{2} \sum_{t=1}^T a'_i(t) \sum_{j=1}^n (a_j(t-1) + a_j(t))$$ to be the cost averaging the permanent impact contribution from the previous and current time step. 

\begin{lemma}\label{lem:perm-avg-zero-sum}
    The variant of permanent impact cost $\permmean{}$ satisfies the following: for any action profile $\ba \in \cA(V_1) \times...\times \cA(V_n) $, we have:
    \[
    \sum_{i=1}^n \permmean{}(a_i,a_{-i}) = \frac{1}{2} \left( \sum_{i=1}^n  V_i \right)^2
    \]
\end{lemma}
\ifarxiv
\begin{proof}
    By expanding out $a_i'(t)$ for every $i$ and rearranging the summations, we compute:
    \begin{align*}
        \sum_{i=1}^n \permmean{}(a_i,a_{-i})
        &= \sum_{i=1}^n \frac{1}{2} \sum_{t=1}^T a'_i(t) \sum_{j=1}^n (a_j(t-1) + a_j(t)) \\
        &= \frac{1}{2} \sum_{t=1}^T \sum_{i=1}^n  (a_i(t) - a_i(t-1)) \sum_{j=1}^n (a_j(t-1) + a_j(t)) \\
        &= \frac{1}{2} \sum_{t=1}^T \left(\sum_{i=1}^n  a_i(t) - \sum_{i=1}^n a_i(t-1)\right) \left(\sum_{j=1}^n (a_j(t-1) +  \sum_{j=1}^n a_j(t)\right) \\
        &= \frac{1}{2} \sum_{t=1}^T \left(\sum_{i=1}^n  a_i(t) - \sum_{i=1}^n a_i(t-1)\right) \left(\sum_{i=1}^n (a_i(t-1) +  \sum_{i=1}^n a_i(t)\right)\\
        &= \frac{1}{2} \sum_{t=1}^T \left( \left(  \sum_{i=1}^n  a_i(t) \right)^2 - \left(  \sum_{i=1}^n  a_i(t-1) \right)^2 \right)
    \end{align*}
    where the second-to-last step switches the indexing notation and the last step follows from the identity $(a-b)(a+b)=a^2-b^2$. Now, expanding out the telescoping sum, this quantity equals:
    \begin{align*}
        \frac{1}{2} \left( \left( \sum_{i=1}^n  a_i(T) \right)^2 - \left( \sum_{i=1}^n  a_i(0) \right)^2 \right) = \frac{1}{2} \left( \sum_{i=1}^n  V_i \right)^2
    \end{align*}
    by the boundary conditions $a_i(0)=0$ and $a_i(T) = V_i$ for all $i$. 
\end{proof}
\fi

\subsection{Decomposition Theorem} \label{subsec:decomp}

Finally, we show that the general cost of trading $c$ can be written as a weighted sum of the temporary impact cost $\temp{}$ and the time-averaged variant of permanent impact cost $\permmean{}$. 
\ifarxiv
This implies that the general setting is a mixture of a potential game---coinciding with the temporary impact regime---and a constant-sum game---coinciding with (roughly) the permanent impact regime.
\fi

\begin{theorem}\label{thm:decomp}
    Fix a market impact coefficient $\kappa$. Then, the cost of trading can be written as:
    \[
    c(a_i, a_{-i}) = \left(1-\frac{\kappa}{2}\right)\cdot\temp{}(a_i, a_{-i}) + \kappa\cdot\permmean{}(a_i, a_{-i})
    \]
    In particular, the classes of potential games and constant-sum games are closed under scalar multiplication, and so by Theorem \ref{thm:temp-potential} and Lemma \ref{lem:perm-avg-zero-sum}, the terms in the decomposition correspond to a potential game and a constant-sum game. 
\end{theorem}
\ifarxiv
\begin{proof}
    We compute:
    \begin{align*}
        c(a_i, a_{-i}) &= \temp{}(a_i,a_{-i}) + \kappa \cdot \perm{}(a_i,a_{-i}) \\
        &= \sum_{t=1}^T a'_i(t) \sum_{j=1}^n a'_j(t) + \kappa \sum_{t=1}^T a'_i(t) \sum_{j=1}^n a_j(t-1)  \\
        &= \sum_{t=1}^T a'_i(t) \sum_{j=1}^n a'_j(t) \\ & \ \ \ + \sum_{t=1}^T \left(\frac{\kappa}{2} a'_i(t) \sum_{j=1}^n a_j(t-1) + \frac{\kappa}{2} a'_i(t) \sum_{j=1}^n a_j'(t) - \frac{\kappa}{2} a'_i(t) \sum_{j=1}^n a_j'(t) + \frac{\kappa}{2} a'_i(t) \sum_{j=1}^n a_j(t-1) \right) \\
        &= \sum_{t=1}^T a'_i(t) \sum_{j=1}^n a'_j(t) + \sum_{t=1}^T \left(\frac{\kappa}{2} a'_i(t) \sum_{j=1}^n a_j(t-1) + \frac{\kappa}{2} a'_i(t) \sum_{j=1}^n a_j(t) - \frac{\kappa}{2} a'_i(t) \sum_{j=1}^n a'_j(t) \right) \\
        &= \left(1-\frac{\kappa}{2}\right) \sum_{t=1}^T a'_i(t) \sum_{j=1}^n a'_j(t) + \frac{\kappa}{2} \sum_{t=1}^T a'_i(t) \sum_{j=1}^n (a_j(t-1) + a_j(t)) \\
        &= \left(1-\frac{\kappa}{2}\right)\cdot\temp{}(a_i,a_{-i}) + \kappa\cdot\permmean{}(a_i,a_{-i})
    \end{align*}
    as desired. 
\end{proof}
\fi




\section{Efficient Equilibria Computation in the General Game}\label{sec:no-regret-dynamics}

We have just established that the general trading game admits a meaningful decomposition; however, this has no immediate implications for tractable equilibria computation. In this section, we relax our goal to the computation of CCE, a broader class of equilibria that encapsulates Nash equilibria. We show how to efficiently implement \textit{no-regret dynamics}, for which the empirical history of play converges to a joint distribution that is an (approximate) CCE. In no-regret dynamics, each player chooses a sequence of (randomized) actions by running a regret-minimizing algorithm over $R$ rounds.

\begin{definition}[Regret]
    The (average) regret of a player $i$ who chooses a sequence of strategies $a_{i,1},...,a_{i,R} \in \cA_i$ is defined as:
    \[
    Reg_i(R) = \max_{a_i\in\cA_i} \frac{1}{R} \sum_{r=1}^R (c(a_{i,r}, a_{-i,r}) - c(a_i, a_{-i,r}))
    \]
    where $a_{-i,r}$ is the profile of strategies chosen by players excluding $i$ at round $r$.
\end{definition}

No regret encodes the idea that if a player looks back at the history of play, they would find that no deviation to a fixed strategy would have improved their cost. It thus follows that if each player has regret bounded by $\eps$, the time-averaged, empirical distribution of play is an $\eps$-approximate CCE. 

The difficulty in our setting is producing a no-regret algorithm that is computationally efficient, given that the space of trading strategies is exponentially large. 
\ifarxiv
A classic no-regret algorithm is the Multiplicative Weights algorithm (see~\citep{AroraMW} for an extensive survey), which achieves regret decreasing at a rate of $O(\sqrt{\ln|\cA|/R})$. Thus, if all players run their own Multiplicative Weights algorithm, the empirical distribution of play converges to an $\eps$-approximate CCE after $O(\ln|\cA|/\eps^2)$ rounds. 
\fi
However, running Multiplicative Weights will be computationally expensive for our problem: it maintains an explicit distribution over the strategy space and so requires per-round computation that is linear in the number of strategies---which, for our setting, is exponential in $T$. 

We show that an instantiation of the Follow The Perturbed Leader (FTPL) algorithm \citep{kalai03efficient} can be used to implement no-regret dynamics efficiently---the number of rounds required to reach an approximate equilibrium and the per-round runtime is polynomial in our problem parameters. We note that FTPL maintains no-regret guarantees in adversarial environments, so while our main focus is on the joint execution of FTPL by all players, our instantiation of FTPL can be used by a \textit{single} player to obtain vanishing regret in \textit{any} environment, where other players could be acting adversarially. 
\ifarxiv
Next we introduce FTPL and state its regret guarantees.
\fi

\ifarxiv
\paragraph{FTPL preliminaries.} 
FTPL (Algorithm \ref{alg:ftpl}) is a no-regret algorithm for \textit{online linear optimization} (OLO) problems, defined by a learner with strategy space $\cF\subseteq\R^d$ and an adversary with strategy space $\cH\subseteq\R^d$. At every round $r \in [R]$, the learner chooses a strategy $f_r\in\Delta\cF$ and the adversary chooses a strategy $h_r\in\cH$. The learner then observes $h_r$ and incurs cost $\<f_r, h_r\>$. 

For every round $r$, let $H_r = \sum_{s=1}^{r-1} h_s$ be the cumulative cost observed so far. To run FTPL, the learner best responds to a noisy version of $H_r$. In particular, the learner optimizes over the cost $\<f, H_r+N_r\>$, where $N_r$ is chosen uniformly at random from $[0,\eta]^d$. Rephrased, the learner samples from a distribution given by the randomized best response. Crucially, this distribution is maintained only implicitly (unlike e.g. Multiplicative Weights); for low-dimensional problems, then, the brunt of the computation lies in the optimization/best response step.

\begin{theorem}\citep{kalai03efficient}\label{thm:ftpl}
    Let $D = \max_{f,f'\in\cF}\|f-f'\|_1, M =\max_{h\in\cH}\|h\|_1,$ and $C = \max_{f\in\cF,h\in\cH} |\<f, h\>|$. Against any adversary's choice of strategies $h_1,...,h_R$, FTPL (Algorithm \ref{alg:ftpl}) with noise parameter $\eta = \sqrt{\frac{2MCR}{D}}$ obtains regret:
    $$
    \max_{f\in\cF} \frac{1}{R} \sum_{r=1}^R \left( \E[\<f_r, h_r\>] - \<f, h_r)\> \right) \leq 2\sqrt{\frac{DMC}{R}}
    $$
    where the expectation is taken over the noise vectors.
\end{theorem}

\paragraph{An OLO formulation of the trading game.}
In order to implement FTPL, we must first cast our problem as an instance of OLO. We will take the perspective of player $i$ (who corresponds to the learner) computing a no regret strategy against other players (who, in aggregate, correspond to the adversary). Although players' cost functions are \textit{not} linear in their actions, we show how to ``linearize" the problem by constructing higher-dimensional representations of the strategy spaces for the learner and adversary. Broadly speaking, we use the fact that the cost function is linear in the opponents' strategies $a_{-i}$ and introduce dimensions to represent nonlinearities in $a_i$---in particular the product relationships. More specifically, the learner will play over the space $\cF_{\cA_i} = \{f(a_i)\}_{a_i\in\cA_i} \subseteq \R^{2T}$ and the adversary will play over the space $\cH_{\cA_{-i}} = \{h(a_{-i})\}_{a_{-i}\in \cA_{-i}} \subseteq \R^{2T}$, where $h$ and $f$ apply the following transformations:
\[
f(a_i) = \begin{bmatrix}
a'_i(1) \\
\vdots \\
a'_i(t) \\
\vdots \\
a'_i(T) \\
a'_i(1)(a'_i(1) + \kappa a_i(0)) \\
\vdots \\
a'_i(t)(a'_i(t) + \kappa a_i(t-1)) \\
\vdots \\
a'_i(T)(a'_i(T) + \kappa a_i(T-1))
\end{bmatrix} 
\hspace{1em}, \hspace{1em}
h(a_{-i}) = 
\begin{bmatrix}
\sum_{j\neq i} a'_j(1) + \kappa a_j(0) \\
\vdots \\
\sum_{j\neq i} a'_j(t) + \kappa a_j(t-1) \\
\vdots \\
\sum_{j\neq i} a'_j(T) + \kappa a_j(T-1) \\
1 \\
\vdots \\
1
\end{bmatrix}
\]
Note that the dimension of the strategy space is only larger by a factor of 2. The transformation is cost-preserving: the learner's cost of playing $f(a_i)$ against $h(a_{-i})$ in the OLO problem matches their cost of playing $a_i$ against $a_{-i}$ in the trading game:
\begin{align*}
    \<f(a_i), h(a_{-i})\> &= \sum_{t=1}^T a_i'(t) \sum_{j\neq i} (a'_j(t) + \kappa a_j(t-1))  + \sum_{t=1}^T  a'_i(t)(a'_i(t) + \kappa a_i(t-1))\\
    &= \sum_{t=1}^T \left(a'_i(t) \sum_{j=1}^n a'_j(t) + \kappa a'_i(t) \sum_{j=1}^n a_j(t-1) \right)\\
    &= c(a_i, a_{-i})
\end{align*}

\begin{algorithm}[H]
    \For{$r=1$ \KwTo $R$}{
        Let $H_r = \sum_{s=1}^{r-1} h_s$ be the cumulative cost so far\;
        Let $N_r \sim [0, \eta]^{d}$ be a noise vector chosen uniformly at random\;
        Choose the strategy $f_{r} = \argmin_{f\in\cF} \<f, H_r + N_r\>$\;
        Observe $h_r$\;
    }
    \caption{FTPL}
    \label{alg:ftpl}
\end{algorithm}


With this instantiation in hand, we can now appeal to the guarantees of FTPL in the following corollary to Theorem \ref{thm:ftpl}. 
\begin{corollary}\label{cor:ftpl}
    In our instantiation, the quantities $D,M,$ and $C$ are polynomial in $n$, $T$, and $\theta$, where $\theta = \max\{|\theta_L|, |\theta_U|\}$. In particular, we have that $D\leq O(\theta^2 T^2)$, $M\leq O((n-1)\theta T^2)$, and $C\leq O(n\theta^2 T^2)$. Plugging this in, FTPL obtains regret bounded by $O\left(\frac{n\theta^{5/2}T^3}{\sqrt{R}}\right)$ in our instantiation.
\end{corollary}

It remains to consider how to solve the (randomized) best response problem of FTPL in our instantiation. Observe that there is a one-to-one correspondence between strategies $a_i\in\cA_i$ and $f(a_i)\in\cF$, and so we can speak interchangeably about choosing strategies $a_i$ and $f(a_i)$. Thus, we can write the best response problem as finding: 
\[
a_{i,r} = \argmin_{a_i\in\cA_i} \<f(a_i), H_r + N_r\>
\]
Using the notation $v^k$ for the $k^{th}$ coordinate of a vector $v$, observe that we can write:
\begin{align*}
    \<f(a_i), H_r + N_r\> &= \sum_{k=1}^{2T} f(a_i)^k (H_r+N_r)^k \\
    &= \sum_{t=1}^T a_i'(t)(H_r+N_r)^t + \sum_{t=1}^T a_i'(t)(a_i'(t)+\kappa a_i(t-1))(H_r+N_r)^{T+t}
\end{align*}
Notice that once we have fixed $H_r$ and $N_r$, the cost at each step is solely a function of $a_i(t-1)$ and $a_i'(t)$. Thus, we can invoke Algorithm \ref{alg:dp-BR} as a subroutine, instantiated with the one-step cost: $$p^t_r(a_i(t-1), a_i'(t)) \coloneqq  a_i'(t)(H_r+N_r)^t + a_i'(t)(a_i'(t)+\kappa a_i(t-1))(H_r+N_r)^{T+t}$$ 
Then, since computing $H_r$ and $N_r$ at every round can be done in time $2T$, the running time of FTPL directly inherits from the guarantees of Algorithm \ref{alg:dp-BR}.

\begin{corollary}\label{cor:ftpl-runtime}
    Our instantiation of FTPL has per-round running time $O(\theta^2T^2)$, where $\theta = \max\{|\theta_L|, |\theta_U|\}$. 
\end{corollary}

\paragraph{No-regret dynamics.} Finally, to implement no-regret dynamics, every player $i\in[n]$ maintains a copy of FTPL (Algorithm \ref{alg:ftpl}). In rounds $r\in[R]$, every player simultaneously draws a strategy $a_{i,r}$ from the distribution maintained by their copy of FTPL (therefore ensuring that each player's randomness is private). Then, every player observes the full action profile $(a_{1,r},...,a_{n,r})$ and updates their copy of FTPL with the cost vector $h(a_{-i,r})$.

\begin{corollary}\label{cor:ftpl-dynamics}
    For every player $i\in[n]$, let $a_{i,1},...,a_{i,R}$ be draws from the distributions maintained by FTPL in no-regret dynamics, set with noise parameter $\eta = nT\sqrt{2\theta R}$, where $\theta = \max\{|\theta_L|, |\theta_U|\}$. 
    Let $\mathbf{D}$ be the empirical distribution over the realized action profiles $\ba_1,...,\ba_R$, where $\ba_r = (a_{1,r},...,a_{n,r})$. Then, $\mathbf{D}$ is an $\eps$-approximate coarse correlated equilibrium after $R = O\left(\frac{n^2\theta^5T^6}{\eps^2}\right)$ rounds of no-regret dynamics, with total per-round running time $O(n\theta^2 T^2)$.
\end{corollary}
\fi

In Appendix \ref{app:swap-regret}, we investigate computation of approximate CE via no-\textit{swap}-regret dynamics.

\section{Experiments}\label{sec:experiments}
\ifarxiv 
In light of our theoretical results, it is natural to ask how quickly no-regret dynamics converge to an approximate CCE in actual implementation, and what the approximate equilibria look like --- in particular, are they ``close" to the stronger notion of Nash equilibria? We empirically investigate these questions under different regimes of market impact, which dictate how ``close" or ``far'" the game is from a potential game and a zero-sum game.
\fi
The code can be found \href{https://github.com/mirahshi/strategic-trading/}{here}.

\ifarxiv
\paragraph{Parameter settings.} 
\fi
We implement no-regret dynamics between two players using our instantiation of FTPL as described in Section \ref{sec:no-regret-dynamics}. Throughout, we will fix the setting of $T=5, V_1=V_2=10, \theta_L=-5, \theta_U=5$ (thus both players have 5 time periods in which to acquire a net long position of 10 shares, and are able to buy or (short) sell up to 5 shares at each step)
while varying the market impact coefficient $\kappa$. For each setting of $\kappa$ we execute 100 runs of no-regret dynamics, each consisting of 2500 rounds. Recall that FTPL takes in an additional noise parameter $\eta$. Using the theoretical guideline of $\eta \approx \sqrt{\text{number of rounds}}$, we choose $\eta = 50$. 

\ifarxiv
\subsection{Convergence Rate}
First we examine how regret evolves as no-regret dynamics progresses, in order to evaluate the speed of convergence in our implementation. In Figures \ref{fig:reg-cumulative} and \ref{fig:reg-avg}, we show cumulative and average/per-round regret (respectively) as a function of rounds of no-regret dynamics for different settings of $\kappa$. We find that average regret converges to 0 (and so the empirical distribution converges to a coarse correlated equilibria) more rapidly that our theory suggests; while our asymptotic convergence rates scale as $O(T^6/\eps^2)$, we see that average regret (i.e. distance to coarse correlated equilibria) flattens out after 500-1000 rounds for all settings of $\kappa$. 

In Figure \ref{fig:reg-cumulative}, we see that regret behaves somewhat differently over the course of FTPL for different $\kappa$. Most notably, for $\kappa=2$, we see that regret oscillates. As a whole, as $\kappa$ increases, the shape of regret transitions from quickly flattening out, to oscillating, to quickly flattening out again. 
However, the individual trajectories at small $\kappa$ are quite smooth, while behavior becomes
more volatile at larger $\kappa$.
These findings reflect changes in the game's underlying structure --- as we saw in Section \ref{sec:decomposition}, the game morphs from being a potential game (at $\kappa=0$) to a constant-sum game (at $\kappa=2$).

\begin{figure}
\centering
\begin{tabular}{cccc}
& \includegraphics[width=45mm,trim={0mm 11mm 0mm 0},clip]{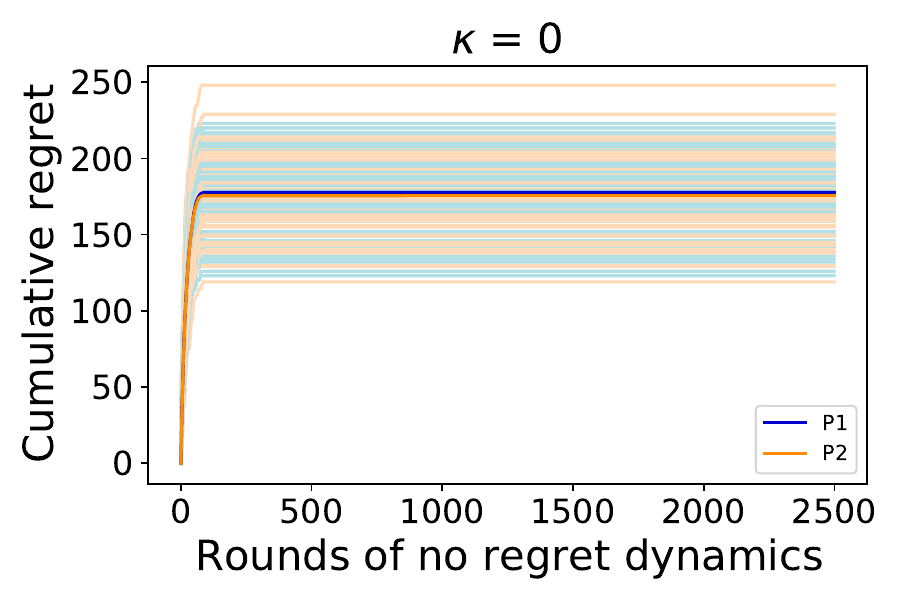} & 
\hspace{2pt}
\includegraphics[width=40mm,trim={11mm 11mm 3mm 0},clip]{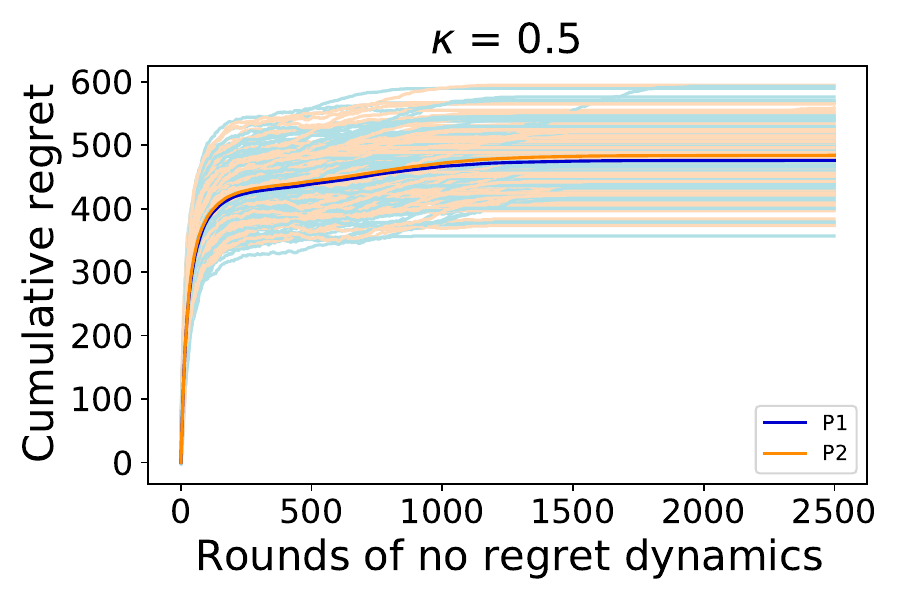}  
& \includegraphics[width=40mm,trim={11mm 11mm 3mm 0},clip]{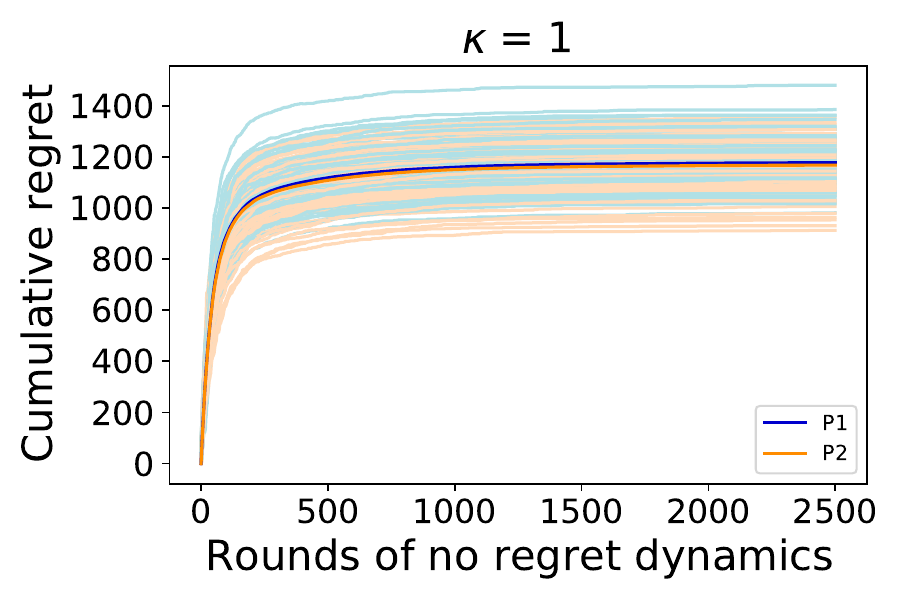} \\ & \includegraphics[width=44mm,trim={0 11mm 3mm 0},clip]{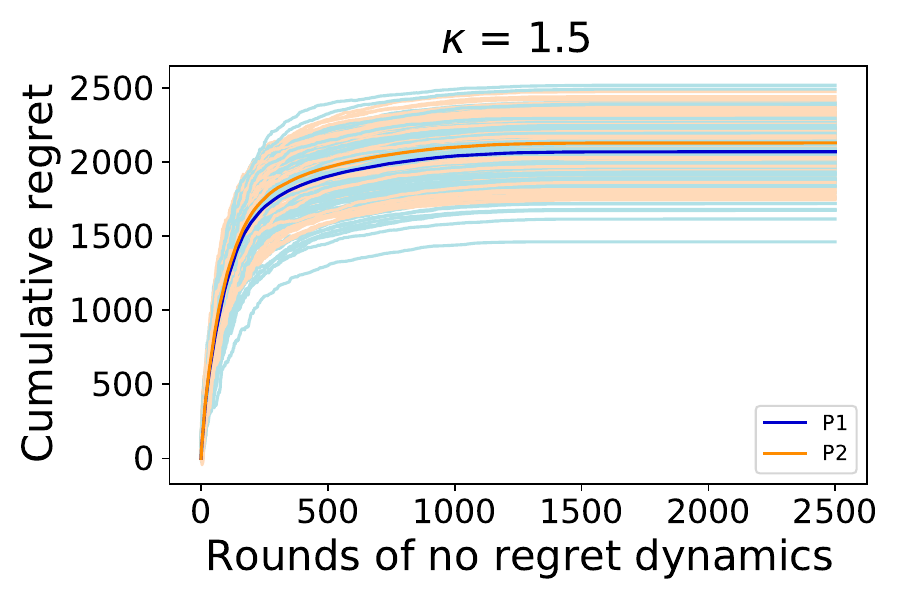} 
& \includegraphics[width=40mm,trim={11mm 11mm 3mm 0},clip]{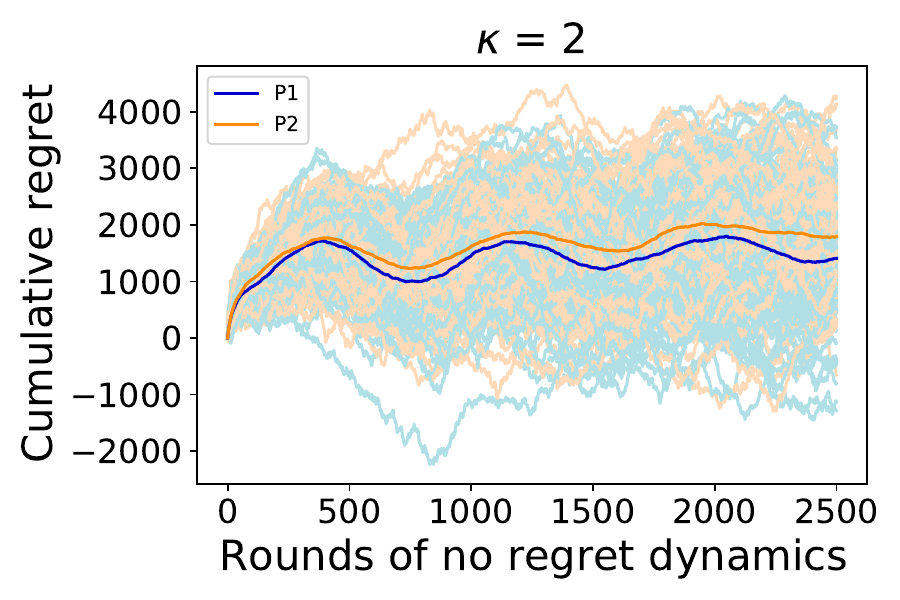} & \includegraphics[width=40mm,trim={11mm 11mm 3mm 0},clip]{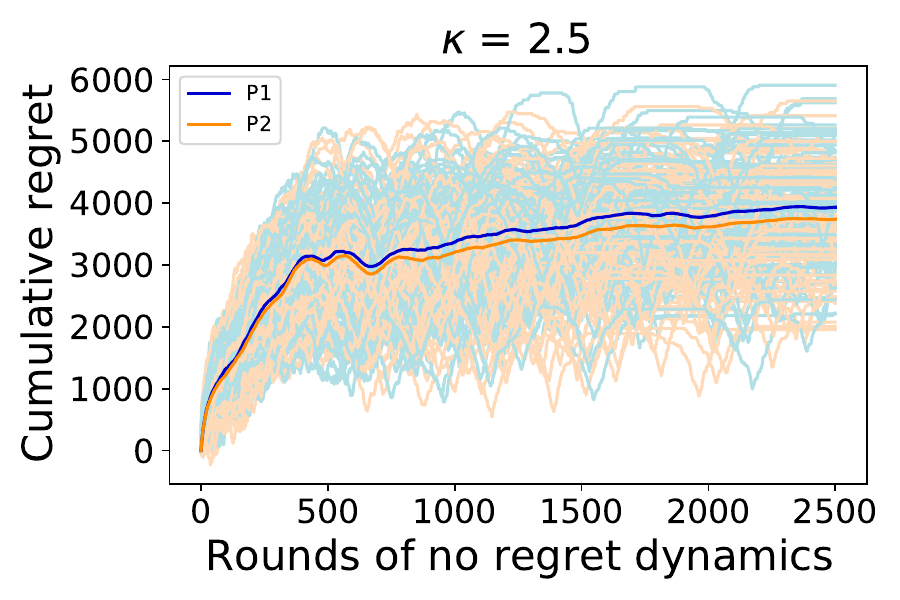} \\
& \includegraphics[width=44mm,trim={0mm 0mm 3mm 0},clip]{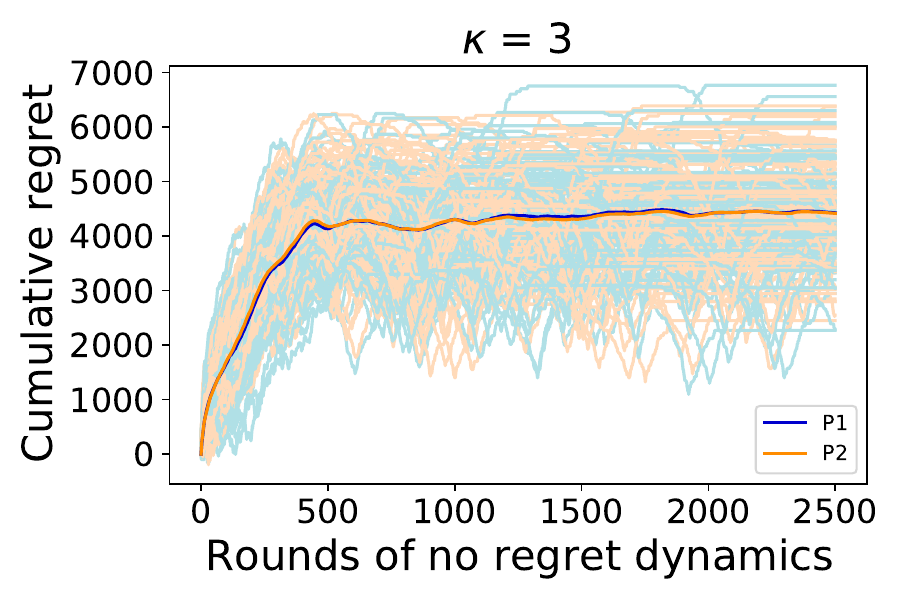} &  \includegraphics[width=40mm,trim={11mm 0 3mm 0},clip]{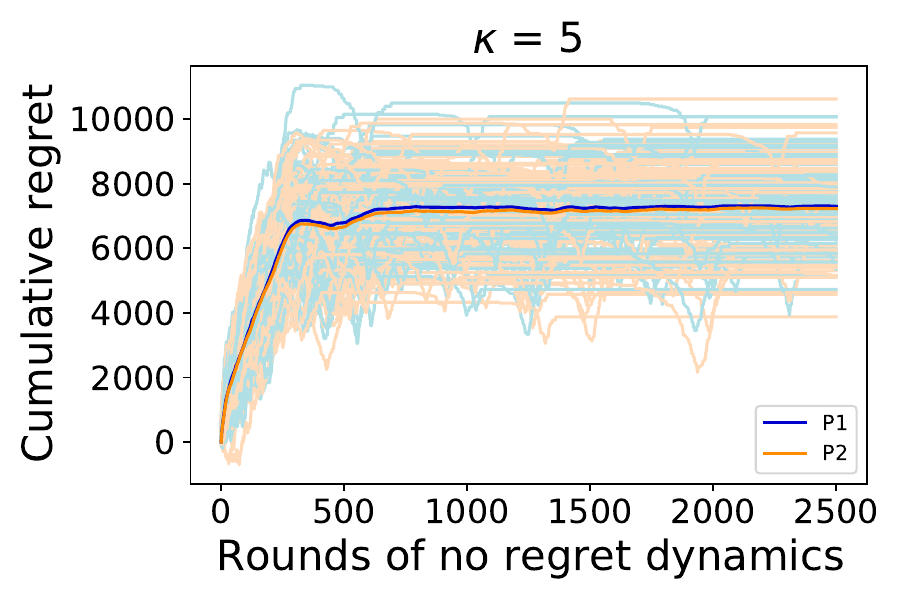} & \includegraphics[width=40mm,trim={11mm 0 3mm 0},clip]{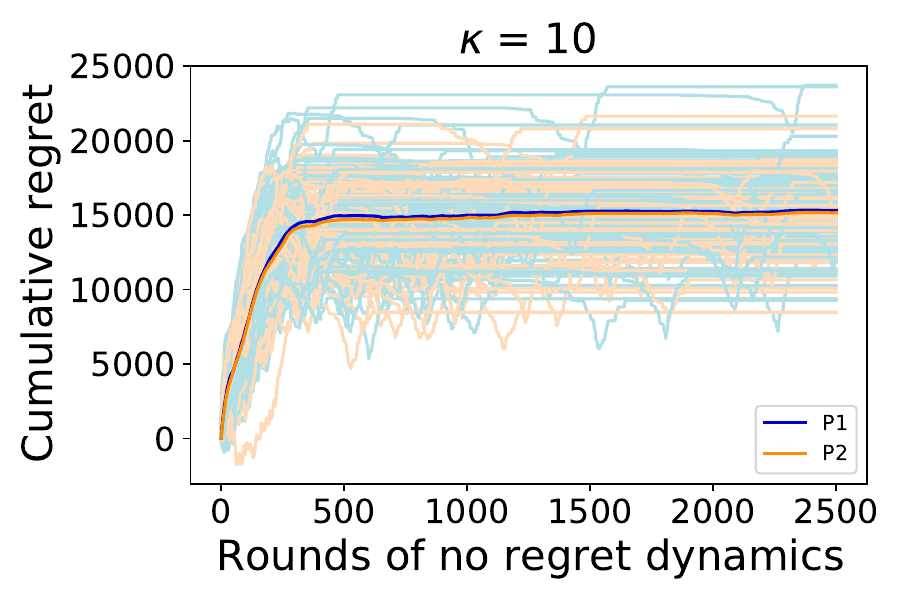} \\
\end{tabular}
\caption{Cumulative regrets of players 1 and 2 for varying $\kappa$, averaged across 100 runs of no-regret dynamics. Faint lines represent individual runs, dark lines represent averages.}
\label{fig:reg-cumulative}
\end{figure}

\begin{figure}
\centering
\begin{tabular}{cccc}
& \includegraphics[width=45mm,trim={0mm 3mm 0mm 0},clip]{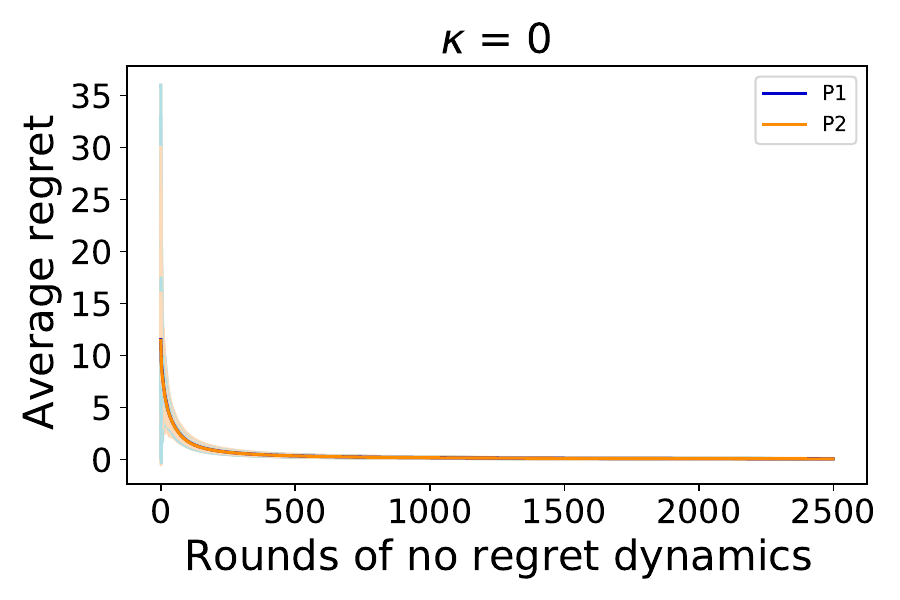} & 
\hspace{2pt}
\includegraphics[width=40mm,trim={11mm 3mm 3mm 0},clip]{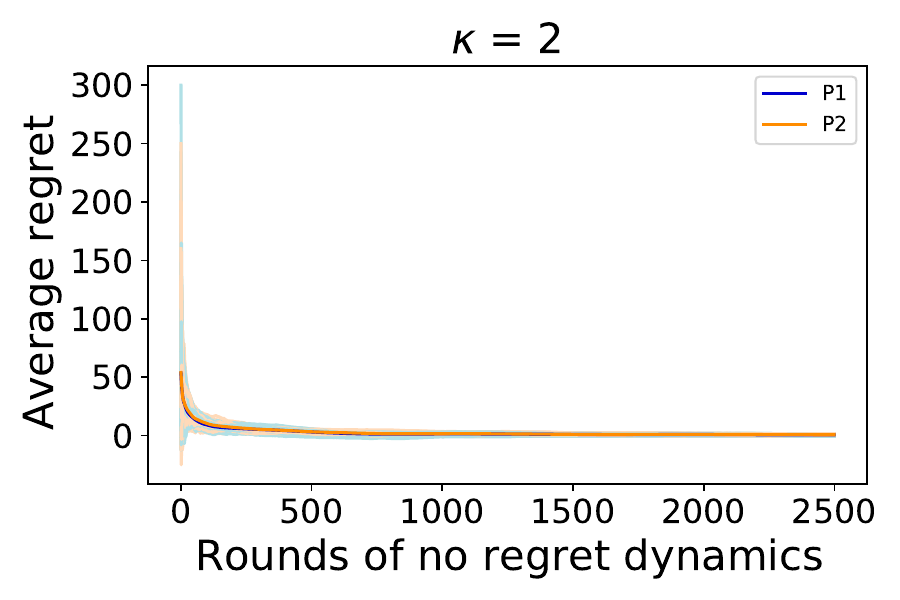}  
& \includegraphics[width=40mm,trim={11mm 3mm 3mm 0},clip]{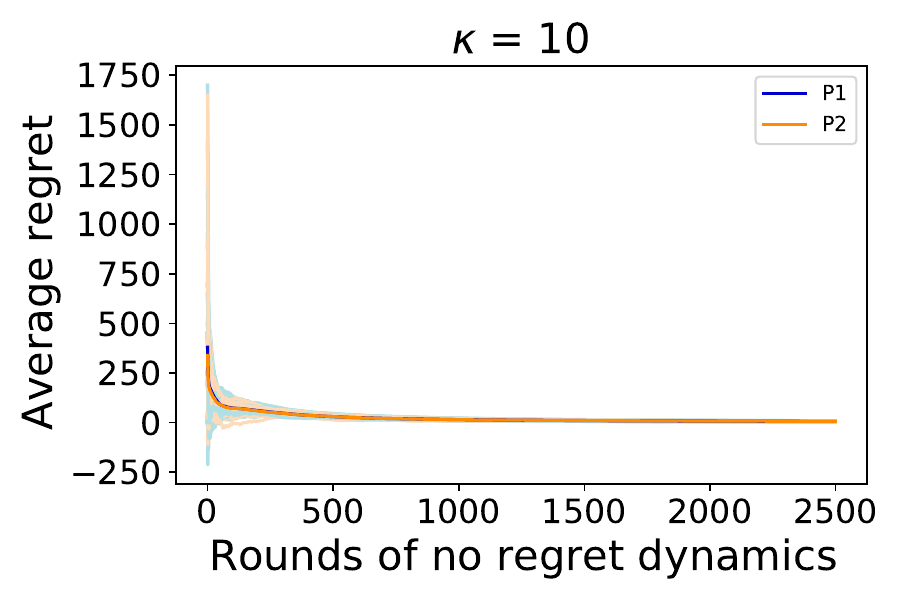} 
\end{tabular}
\caption{Average/per-round regrets of players 1 and 2 for varying $\kappa$, averaged across 100 runs of no-regret dynamics. We exclude other $\kappa$ for concision; the curves (as shown in this manner) look fairly identical.}
\label{fig:reg-avg}
\end{figure}
\fi

\ifarxiv
\subsection{Equilibria Properties}

We have established that no-regret dynamics converge to approximate coarse correlated equilibria fairly quickly. Now we ask: What do the approximate equilibria found by no-regret dynamics look like? Can we say anything more interesting or specific about these equilibria; in particular, are they close to stronger forms of equilibria like Nash equilibria?
\fi

Figure \ref{fig:dists-to-eq} answers the questions: how close is the outputted joint distribution to the stronger notion of CE, and how close are the outputted marginal distributions to a mixed Nash equilibrium? We measure these distances directly using the definitions:
\begin{align*}
    \text{Dist. to NE}&= \max_{i=1,2}  \left[\E_{\substack{a_1\sim D_1 \\ a_{2}\sim D_{2}}}[c(a_i, a_{-i})] - \min_{a^*\in \cA}\E_{a_{-i}\sim D_{-i}}[c(a^*, a_{-i})] \right]\\
    \text{Dist. to CE}&= \max_{i=1,2}  \left[ \E_{(a_1, a_{2})\sim \mathbf{D}}[c(a_i,a_{-i})] - \min_{\phi:\cA\to \cA}\E_{(a_1,a_2)\sim \mathbf{D}}[c(\phi(a_i),a_{-i})] \right]
\end{align*}
\begin{wrapfigure}{l}{0.4\textwidth}
\vspace{-2em}
  \begin{center}
    \includegraphics[width=\linewidth]{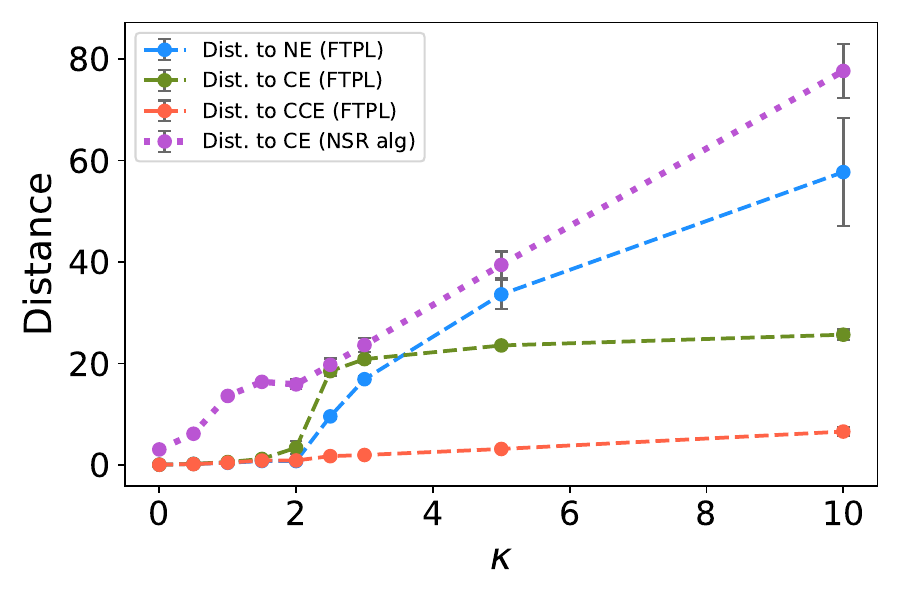}
  \end{center}
  \vspace{-8pt}
  \caption{Distances to Nash equilibria and CE for varying $\kappa$. We use distance to CCE (regret) as a baseline. The plots show means and std. deviations over 100 runs of no-regret dynamics. We also plot the distance to CE found by no-swap-regret dynamics (means and std. deviations over 20 runs of a no-swap-regret algorithm; see Appendix \ref{app:swap-regret}).}
\label{fig:dists-to-eq}
\end{wrapfigure}
Our results show that when the game interpolates between a potential game and zero-sum game ($\kappa\leq 2$), FTPL finds almost-exact CE and mixed Nash equilibria; in fact the distances closely mirror the distances to CCE, which is theoretically guaranteed to be low. 
When the game is zero-sum (i.e. when $\kappa=2$), no-regret dynamics is guaranteed to converge to a mixed Nash equilibrium \citep{Cesa-Bianchi_Lugosi_2006}. Thus for $\kappa\geq 2$, the figure presents an intuitive relationship: as the game becomes less like a zero-sum game (i.e. as $\kappa$ increases beyond $2$), the distance to NE found by FTPL dynamics also increases. Meanwhile, for larger $\kappa$, the distance to CE increases sharply then plateaus. 

Figure \ref{fig:dists-to-eq} also shows the distance to to CE of the joint distribution outputted by a recent no-\textit{swap}-regret algorithm \citep{dagan2024external, peng2024fast}, which guarantees convergence to an approximate CE. We see that although FTPL dynamics only guarantees convergence to approximate CCE, its joint play is in fact closer to CE than an implementation of no-swap-regret dynamics for all $\kappa$. We formally present the no-swap-regret algorithm and give implementation details in Appendix \ref{app:swap-regret}.

\ifarxiv
\begin{wrapfigure}{r}{0.35\textwidth}
\vspace{-2em}
  \begin{center}
    \includegraphics[width=\linewidth]{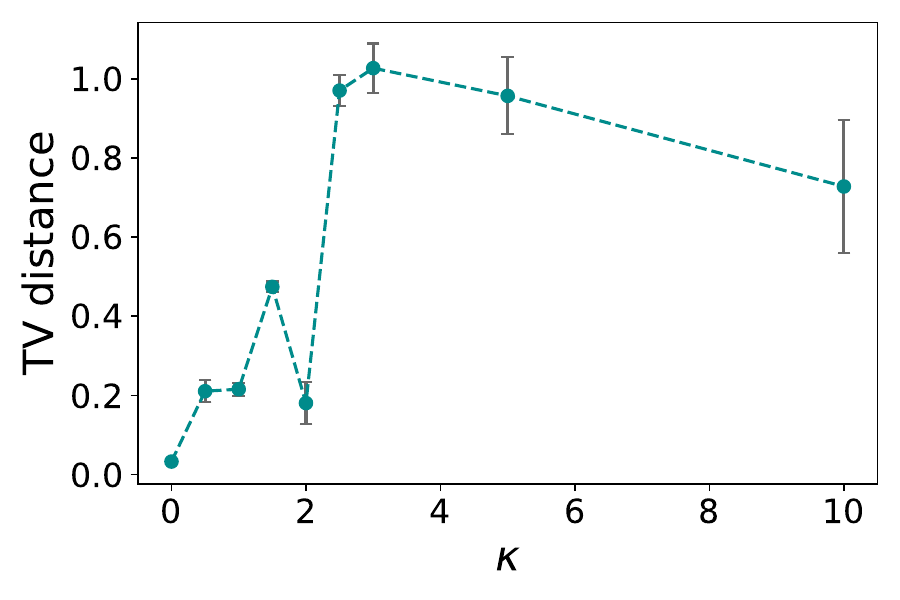}
  \end{center}
  \vspace{-8pt}
  \caption{TV distances between outputted joint distribution and product of marginal distributions, for varying $\kappa$. The plot shows means and std. deviations over 100 runs of no-regret dynamics.}
\label{table:tv-dists}
\end{wrapfigure}
Next we measure how ``correlated" the joint history of play is. Recall that a CCE is a mixed Nash equilibria if its joint distribution can be written as a product distribution---that is, each player's actions can be drawn independently from their own marginal distributions. Since the strategy spaces are not numeric but discrete, combinatorial objects, we cannot measure correlations between player actions in the standard way, but instead will examine
the total variation (TV) distance between the joint distribution returned by no-regret dynamics and the product of each player's marginal distribution. More specifically, let $\mathbf{D}$ be the empirical joint distribution over the realized action pairs $(a_{1,1},a_{2,1}),...,(a_{1,R},a_{2,R})$. Let $D_1$ be the marginal distribution over the first player's actions and $D_2$ be the marginal distribution over the second player's actions. We compute the TV distance between $\mathbf{D}$ and $D_1\times D_2$ as follows:
\[
TV(\mathbf{D},D_1\times D_2) = \sum_{(a_{1},a_{2})\in \text{supp}(\mathbf{D})} \left| \Pr_{\mathbf{D}}[(a_{1},a_{2})] - \Pr_{D_1}[a_1]\cdot\Pr_{D_2}[a_2] \right|
\]
In the sequel, we will refer to this distance informally as ``correlation'' between player strategies.
Figure \ref{table:tv-dists} shows TV distances for varying $\kappa$. For each setting of $\kappa$, we report the average TV distance computed over 100 runs of no-regret dynamics. As expected, TV distance is low for the special case of $\kappa=2$ (no-regret dynamics are known to converge to Nash equilibria in zero/constant-sum games). The TV distance is particularly low for $\kappa=0$ --- in our subsequent results, we see that this can be explained by the fact that when $\kappa=0$, no-regret dynamics finds a \textit{pure} Nash equilibrium fairly quickly (recall that pure Nash equilibria are guaranteed to exist in this regime). For large $\kappa$, the approximate coarse correlated equilibria found by no-regret dynamics exhibit high correlation.
\fi

\ifarxiv
\begin{figure}
\centering
\begin{tabular}{c}
\includegraphics[width=145mm,trim={45mm 8mm 0mm 2mm},clip]{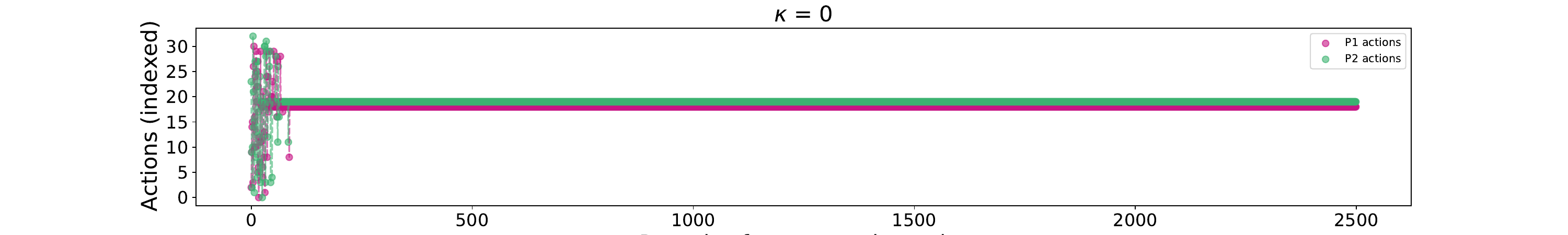} \\  \includegraphics[width=145mm,trim={45mm 8mm 0mm 0},clip]{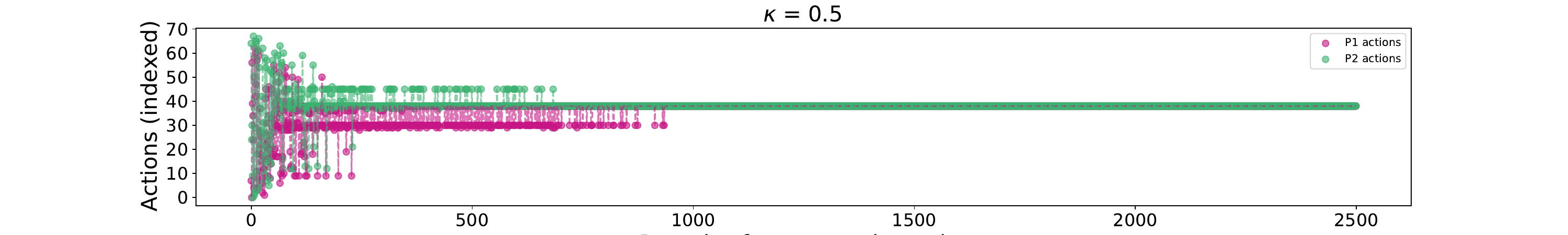} 
\\
\includegraphics[width=145mm,trim={42mm 8mm 0mm 0},clip]{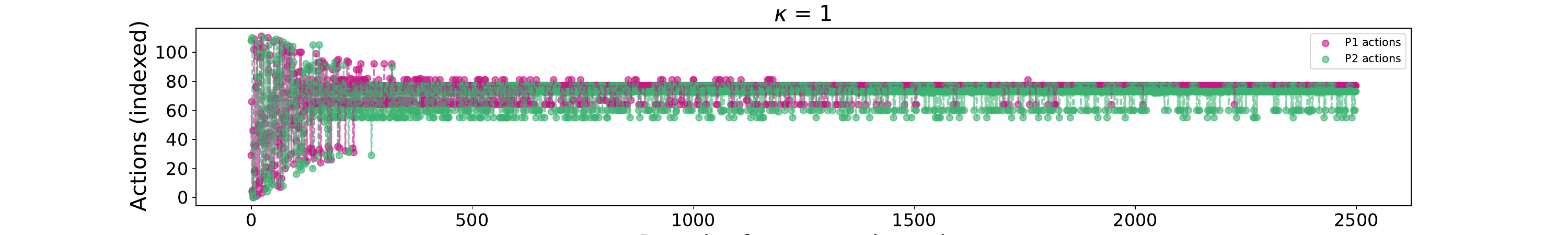}  \\
\includegraphics[width=145mm,trim={42mm 8mm 0mm 0},clip]{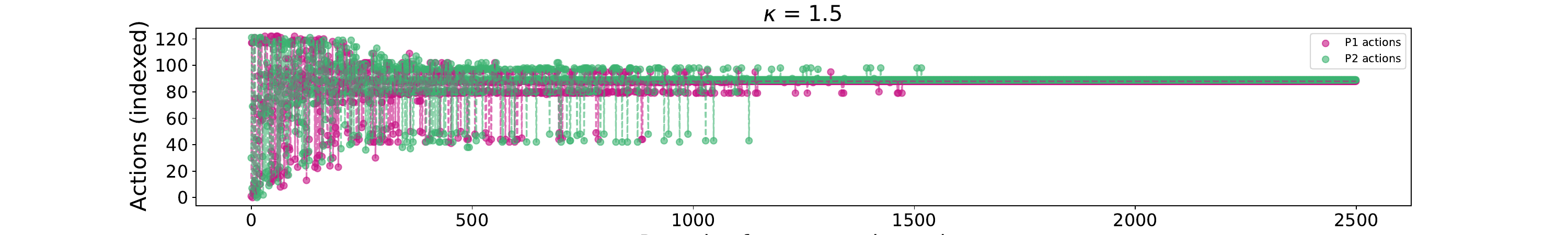}  \\
\includegraphics[width=145mm,trim={45mm 8mm 0mm 0},clip]{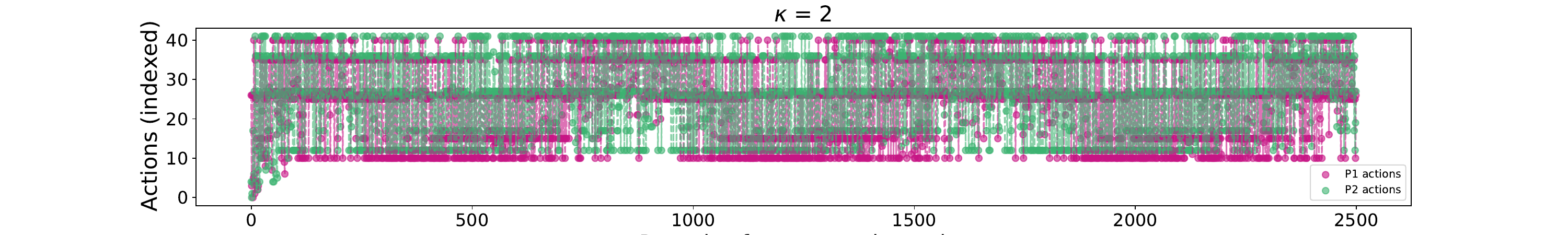}  \\
\includegraphics[width=145mm,trim={45mm 8mm 0mm 0},clip]{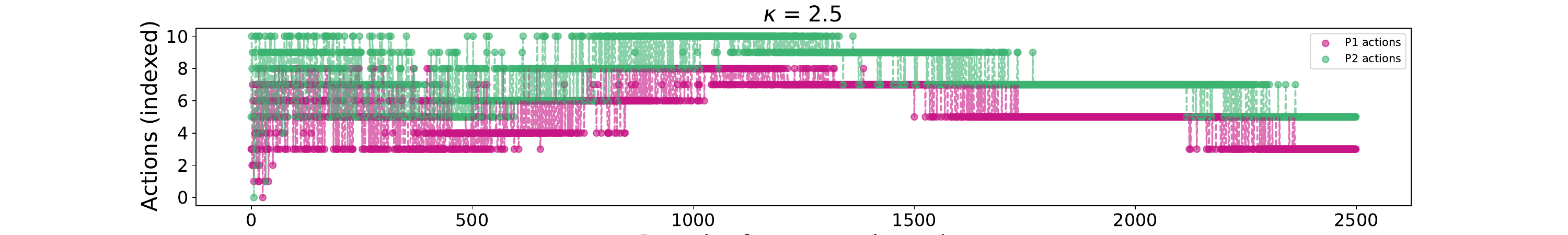}  \\
\includegraphics[width=145mm,trim={45mm 8mm 0mm 0},clip]{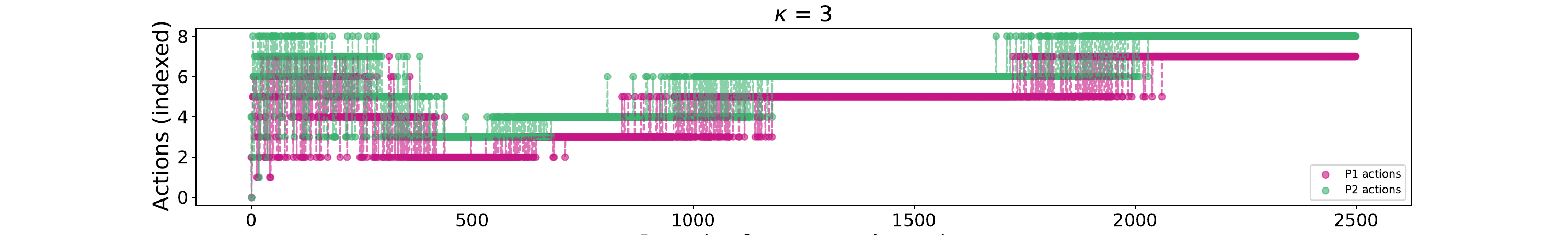}  \\
\includegraphics[width=145mm,trim={45mm 8mm 0mm 0},clip]{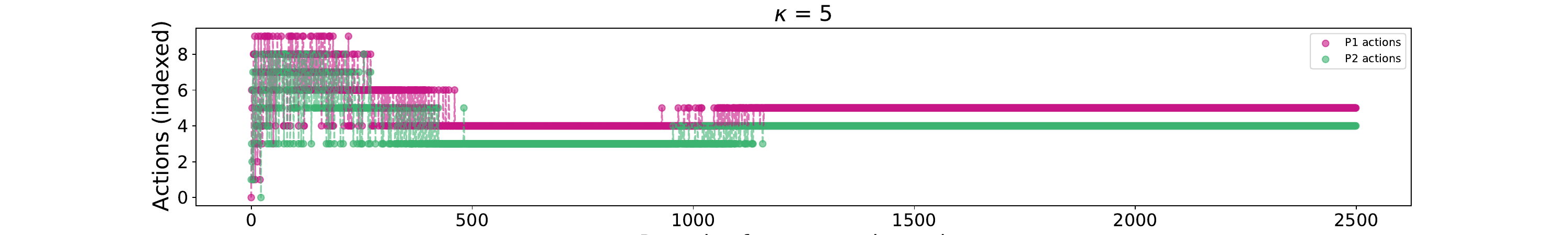}  \\
\includegraphics[width=145mm,trim={45mm 0mm 0mm 0},clip]{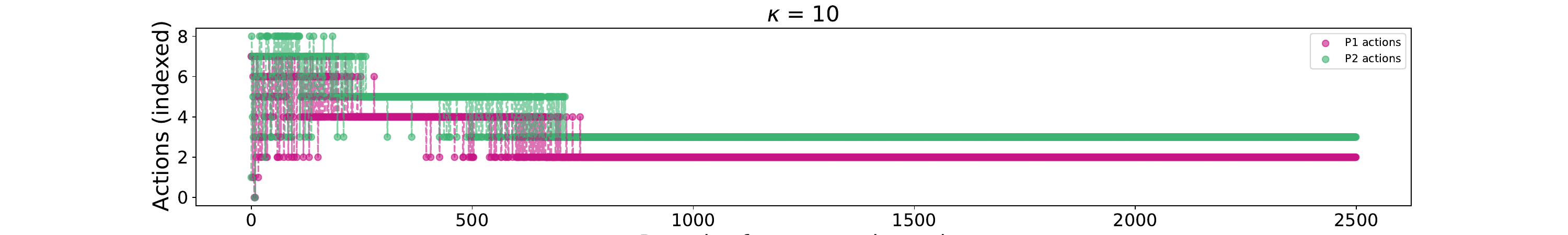}  \\
\end{tabular}
\caption{Actions played by players 1 and 2 over one run of no-regret dynamics for varying $\kappa$ (note: actions are indexed in no particular order and indices might differ from plot to plot; the purpose is to show the progression of actions played)}
\label{fig:actions}
\end{figure}
\fi

\ifarxiv
In Figure \ref{fig:actions}, we take a closer look at the strategy pairs played by FTPL over one run of no-regret dynamics (we note that although we show the outputs of just one run, the behavior is typical across runs). These visualizations help explain some phenomena we find above. First, for small $\kappa$, the equilibria are in fact ``close" to a \textit{pure} Nash equilibrium. Specifically, for $\kappa=0, 0.5,$ and $ 1.5$, no-regret dynamics converges to consistently plays a pure Nash equilibrium after some number of rounds. For $\kappa=0$, this strategy pair is reached fairly quickly. This helps explain why regret stabilizes rapidly and TV distances are low for small $\kappa$. For $\kappa=2$, players oscillate between playing, still, a small subset of actions. For higher $\kappa$, we see oscillatory behavior, albeit longer-lived. We note that for $\kappa=5$ and $10$, the strategy pairs found by the end of 2500 rounds are \textit{not} pure Nash equilibria (even though FTPL might appear to have stabilized), indicating that players might continue to cycle between strategy pairs as no-regret dynamics progresses. 
\fi

\begin{SCfigure}
    \includegraphics[width=55mm,trim={0mm 0mm 0mm 0},clip]{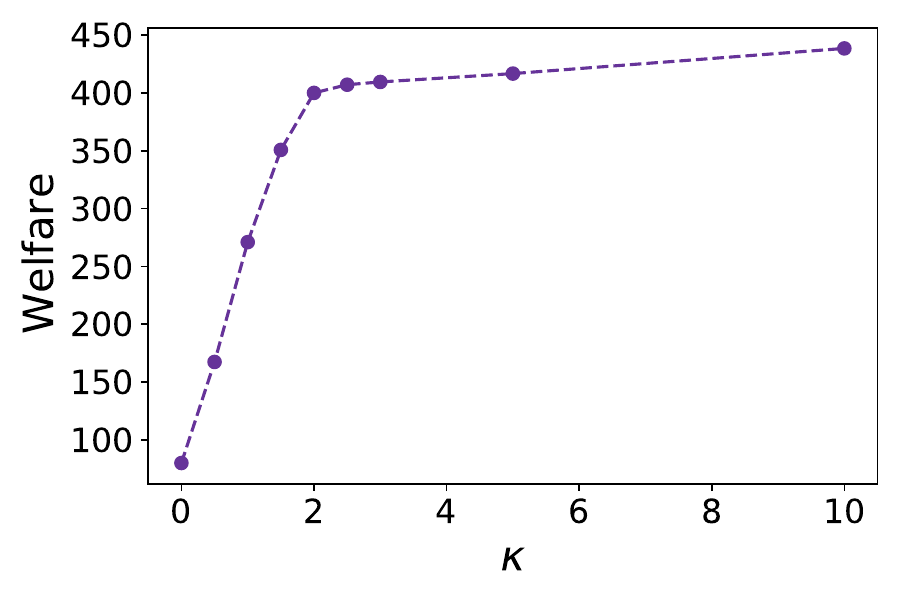}
  \caption{Welfare of approximate CCE for varying $\kappa$. For each $\kappa$, we plot the mean welfare computed over 100 runs of no-regret dynamics (std deviation bars are too small to appear on plot).}
\label{fig:welfare}
\end{SCfigure}

Finally, we touch on the \textit{social welfare} of the equilibria found by no-regret dynamics. Recall that expected welfare is defined as $\sum_{i=1}^n \E_{\ba\in\mathbf{D}}[c(a_i, a_{-i})]$ (since we are talking about costs, \textit{lower} welfare is better).
Figure \ref{fig:welfare} shows the welfare of the approximate coarse correlated equilibria found by no-regret dynamics for varying $\kappa$. In general, we would expect costs --- and thus welfare --- to increase as $\kappa$ increases. Interestingly, we see an inflection point at $\kappa=2$; welfare increases at a slower rate for large $\kappa$. 
\ifarxiv
Previously, we observed that for large $\kappa$, no-regret dynamics returns equilibria exhibiting high correlation (i.e. high TV distances). 
\fi
We can thus interpret this as preliminary evidence showing that allowing for correlation might improve social welfare in markets where permanent impact dominates. Indeed, we can think of correlation as a form of \textit{collusion}. Our findings suggest that as $\kappa$ increases---intuitively, this is also as we move farther from a zero-sum game---the potential benefits of collusion increase. Moreover, this type of collusive behavior can emerge organically from simple learning dynamics.

\section{Limitations and Future Work}

Our work leaves open several avenues for future work. In this work, we considered static strategies. Although we can view our static strategies as baselines or planned trajectories for more sophisticated strategies, future research could explore competitive trading with dynamic strategies that adapt to market activity. Other interesting directions include more general cost functions, such as non-linear functions and those modeling decaying permanent impact, and incorporating different types of (possibly unknown) players with objectives beyond pure position acquisition. Finally, in our game, the parameter $\kappa$ governs how close the game is to a potential game and a zero-sum game. Given that every game can be decomposed into a potential and a zero-sum game, it is broadly relevant to investigate how equilibria properties and computation evolve with shifting game structure, independent of our specific trading context.

\ifarxiv
\section{Acknowledgements}

We give warm thanks to Neil Chriss, Yuriy Nevmyvaka, Andrew Bennett and Anderson Schneider for helpful discussions.
\fi

\newpage
\bibliographystyle{ACM-Reference-Format}
\bibliography{main}

\newpage
\appendix

\section{Correlated Equilibria Computation}\label{app:swap-regret}

In Section \ref{sec:no-regret-dynamics}, we gave an efficient algorithm to compute approximate coarse correlated equilibria via no-regret dynamics. Here, we investigate efficient computation of correlated equilibria (CE), a stronger equilibrium concept than CCE. 

Just as coarse correlated equilibrium is tied to the notion of regret (otherwise referred to as \textit{external regret}), correlated equilibrium is tied to the notion of \textit{swap regret}. 

\begin{definition}[Swap Regret]
    Let $\Phi_i = \{\phi:\cA_i\to\cA_i\}$ be the collection of all functions mapping actions to actions. The (average) swap regret of a player $i$ who chooses a sequence of actions $a_{i,1},...,a_{i,R} \in \cA_i$ is defined as:
    \[
    SReg_i(R) = \max_{\phi\in\Phi_i} \frac{1}{R} \sum_{r=1}^R (c(a_{i,r}, a_{-i,r}) - c(\phi(a_{i,r}), a_{-i,r}))
    \]
    where $a_{-i,r}$ is the profile of strategies chosen by players excluding $i$ at round $r$.
\end{definition}

No swap regret is a stronger guarantee than no (external) regret; it asks that given the history of play, a player has no incentive to deviate to a fixed strategy \textit{conditioned on the strategy they chose}. By definition, if every player at swap regret at most $\eps$, then the empirical joint distribution of play is an $\eps$-approximate correlated equilibrium. 

To implement no-\textit{swap}-regret dynamics, we use an existing reduction transforming any no-regret algorithm to a no-swap-regret algorithm. The classical reduction of Blum and Mansour \citep{blum07external} guarantees, against any sequence of opponent actions, $SReg(R) \leq |\cA|\cdot Reg(R)$ given an algorithm that obtains regret bounded by $Reg(R)$ (for certain no-regret algorithms, a tighter analysis improves the dependence on $|\cA|$ by a factor of $\sqrt{|\cA|}$). To handle a large action space, we use recent reductions of \citet{dagan2024external} and \citet{peng2024fast} that guarantee vanishing swap regret at a rate depending only on the external regret guarantee of the no-regret algorithm, avoiding any dependence on the number of actions (however this comes at an exponential cost in the approximation parameter). We state its guarantees below but defer details of the reduction to \citet{dagan2024external} and \citet{peng2024fast}.

\begin{theorem}[Theorem 3.1 of \citet{dagan2024external}]\label{thm:treeswap}
    Fix an action set $\cA$. Fix $M,d, R \in\mathbb{N}$ such that $M^{d-1}\leq R \leq M^d$.   
    Given an algorithm that, against any adversarial sequence of actions, guarantees regret at most $Reg(R)$ after $R$ rounds, there is an algorithm producing randomized actions $p_{1},...,p_{R} \in\Delta\cA$ such that (in expectation over the randomized actions):
    \[
    SReg(R) \leq Reg(M) + \frac{3}{d}
    \]
    Moreover, if the per-round running time of the no-regret algorithm is $C$, then the per-round amortized running time of this algorithm is $O(C)$.
\end{theorem}

Importing our instantiation of FTPL, we obtain the following guarantee on the joint history given by no-swap-regret dynamics. The number of rounds needed to reach an $\eps$-approximate correlated equilibrium is polynomial in the parameters of the game, but depends exponentially on $1/\eps$. 

\begin{corollary}\label{cor:treeswap}
    Fix a player $i$ in the trading game with action set $\cA_i$. Using the instantiation of FTPL given by Corollary \ref{cor:ftpl}, there is an algorithm producing randomized actions $p_{1},...,p_{R} \in\Delta\cA_i$ such that, in expectation over the randomized actions, $SReg_i(R)\leq \eps$ after $R=O\left(\left(\frac{n^2\theta^5T^6}{\eps^2}\right)^{\frac{1}{\eps}}\right)$ rounds. Moreover, the per-round amortized running time is $O(\theta^2T^2)$, where $\theta = \max\{|\theta_L|, |\theta_U|\}$. If every player $i$ runs a copy of this algorithm, the empirical distribution $\mathbf{D}$ is an $\eps$-approximate correlated equilibrium after $R=O\left(\left(\frac{n^2\theta^5T^6}{\eps^2}\right)^{\frac{1}{\eps}}\right)$ rounds of no-swap-regret dynamics. 
\end{corollary}
\begin{proof}
    We plug the regret guarantee of FTPL from Corollary \ref{cor:ftpl}, $M=\frac{n^2\theta^5T^6}{\eps^2}$, and $d=\frac{1}{\eps}$ into Theorem \ref{thm:treeswap}. The runtime complexity follows from Corollary \ref{cor:ftpl-runtime} and Theorem \ref{thm:treeswap}.
\end{proof}


\subsection{Details of Experimental Implementation }\label{app:swap-regret-experiments}

We implement the no-swap-regret algorithm of Corollary \ref{cor:treeswap} and plot the distance to CE (i.e. swap regret) of the joint play in Figure \ref{fig:dists-to-eq}. In our implementation, we use the same game parameters as before: $T=5, V_1=10, V_2=10, \theta_L=-5, \theta_U=5$. Given that $M\gg d$ in theory, we set $M = 150$ and $d=2$. For each setting of $\kappa$, we execute 20 runs of no-swap-regret-dynamics, each consisting of $M^d = 22,500$ rounds. 

The results suggest that in practice, the algorithm is indeed hindered by the (exponentially) slow convergence rate given by theory. In fact, we find that FTPL obtains \textit{lower} swap regret than this algorithm (it only guarantees low external regret), even with \textit{much} fewer rounds. 

\end{document}